%% file: universal_submodular.tex
\documentclass[]{article}

\usepackage{amsthm,amsmath,amssymb}
\usepackage[linesnumbered,ruled,vlined]{algorithm2e}
\usepackage{authblk}
\usepackage{url}
\usepackage[margin=1in]{geometry}
\usepackage{times}
\usepackage{newtxtext,newtxmath}
\usepackage{multirow}
\usepackage{color}

\newcommand{\argmax}{\mathop{\rm arg\,max}}
\newcommand{\argmin}{\mathop{\rm arg\,min}}
\newcommand{\supp}{\mathop{\rm supp}}
\newcommand{\ALG}{{\rm ALG}}
\newcommand{\OPT}{{\rm OPT}}
\newcommand{\sv}{single-valuable\xspace}
\newcommand{\SV}{S}

\newcommand{\ot}{\leftarrow}

\newcommand{\app}{\gamma}

\newcommand{\Rset}{\mathbb{R}}

\newcommand{\submax}{{\sc SMPSC}\xspace}
\newcommand{\knapsack}{{\sc KPSC}\xspace}
\newcommand{\submaxUC}{{\sc SMPUC}\xspace}
\newcommand{\knapsackUC}{{\sc KPUC}\xspace}

\usepackage{version}
\excludeversion{conference}
\includeversion{fullpaper}

\newcommand{\para}[1]{
\medskip
\noindent
{\bfseries #1:}\hspace{-0.2em}
}

\newtheorem{theorem}{Theorem}
\newtheorem{lemma}{Lemma}
\newtheorem{claim}{Claim}

\date{}

\title{Submodular maximization with uncertain knapsack capacity}

\author{Yasushi Kawase\\
Tokyo Institute of Technology, Tokyo, Japan.
  Email: kawase.y.ab@m.titech.ac.jp
\and
Hanna Sumita\\
  National Institute of Informatics, Tokyo, Japan. JST, ERATO,
  Kawarabayashi Large Graph Project.
  Email: sumita@nii.ac.jp
\and 
Takuro Fukunaga\\
  RIKEN Center for Advanced Intelligence Project, Tokyo, Japan.\\
  Email: takuro.fukunaga@riken.jp
}

\begin{document}

\maketitle

\begin{abstract}
We consider the maximization problem of monotone submodular functions under an uncertain knapsack constraint. Specifically, the problem is discussed in the situation
that the knapsack
 capacity is not given explicitly and can be accessed only through an
 oracle that answers whether or not the current solution is feasible when an
 item is added to the solution.
 Assuming that cancellation of the last item is allowed when it overflows the
 knapsack capacity,
 we discuss the robustness ratios of adaptive policies for this problem, 
 which are the worst case ratios of the objective values achieved by the output solutions
 to the optimal objective values.
 We present a randomized policy of robustness ratio $(1-1/e)/2$,
 and a deterministic policy of robustness ratio $2(1-1/e)/21$.
 We also consider a universal policy that chooses items following a
 precomputed sequence. 
 We present a randomized universal policy of robustness ratio
 $(1-1/\sqrt[4]{e})/2$.
 When the cancellation is not allowed, no randomized adaptive policy
 achieves a constant robustness ratio. Because of this hardness,
 we assume that a probability distribution of the knapsack capacity
 is given, and consider computing a sequence of items
 that maximizes the expected objective value. We present a polynomial-time
 randomized algorithm of approximation ratio
 $(1-1/\sqrt[4]{e})/4-\epsilon$ for any small constant
 $\epsilon >0$.
 \end{abstract}

\input{intro.tex}
\input{preliminaries.tex}
\input{universal.tex}
\input{upperbound.tex}
\input{stochastic.tex}

\input{conclusion.tex}

\para{Acknowledgement}
The first author is supported by JSPS KAKENHI Grant Number JP16K16005. 
The second author is supported by JSPS KAKENHI Grant Number JP17K12646 and JST ERATO Grant Number JPMJER1201, Japan. 
The third author is supported by
JSPS KAKENHI Grant Number JP17K00040 and 
JST ERATO Grant Number JPMJER1201, Japan.
 
\bibliography{knapsack}

  \end{document}

%% file: intro.tex
 \section{Introduction}
The submodular maximization is one of the most well-studied 
combinatorial optimization problems. 
Since it captures an essential part of decision-making situations,
it has a huge number of applications
in diverse areas of computer science.
Nevertheless, 
the standard
setting of the submodular maximization problem fails to capture
several realistic situations.
For example, let us consider 
choosing several items to maximize
a reward represented by a submodular function
subject to a resource limitation.
When the amount of the available resource is exactly known, 
this problem 
 is formulated as the submodular maximization problem with a knapsack constraint.
 However, in many practical cases, precise information on the available resource
is not given. Thus, algorithms for the standard submodular maximization problem
cannot be applied to this situation.
Motivated by this fact, we study 
robust maximization of submodular functions with an uncertain knapsack capacity.
Besides the practical applications,
it is interesting to study this problem because it shows how much
robustness can be achieved for an uncertain knapsack capacity in submodular maximization.

More specifically,
we study the \emph{submodular maximization problem with an unknown knapsack capacity} (\submaxUC).
In this problem, we are given a set $I$ of items, and a
monotone nonnegative submodular function $f\colon 2^I \rightarrow \Rset_+$ such that $f(\emptyset)=0$,
where each item $i \in I$ is associated with a size $s(i)$.
The objective is to find a set of items that maximizes the submodular
function subject to a knapsack constraint, 
but we assume that the knapsack capacity is unknown.
We have access to the knapsack capacity through an oracle;
we add items to the knapsack one by one, and we see 
whether or not the selected items violates the knapsack constraint only after the addition.
If a selected item fits the knapsack, the selection of this item is irrevocable.
When the total size of the selected items exceeds the capacity, there are two settings 
according to whether or not the last selection can be canceled. 
If the cancellation is allowed, then we remove the last selected item from the knapsack,
and continue adding the remaining items to the knapsack. 
In the other setting, we stop the selection, and the final output is defined as the item set in the knapsack before adding the last item.

For the setting where the cancellation is allowed, we consider an \emph{adaptive policy}, which is defined as a decision tree to decide which item to pack into the knapsack next. 
A \emph{randomized} policy is a probability distribution over adaptive policies. 
Performance of an adaptive policy is evaluated by the robustness ratio defined as follows. 
For any number $C \in \Rset_+$, let $\OPT_C$ denote the optimal item set when the capacity is $C$, and let $\ALG_C$ denote an output of the policy. 
Note that if the policy is a randomized one, then $\ALG_C$ is a random variable. 
We call the adaptive policy \emph{$\alpha$-robust}, for some $\alpha \leq 1$, if for any $C \in \Rset_+$, the expected objective value of the policy's output is within a ratio $\alpha$ of the optimal value, i.e., $\mathbb{E}[f(\ALG_C)] / f(\OPT_C) \geq \alpha$. 
We also call the ratio $\alpha$ the \emph{robustness ratio} of the policy.

One main purpose of this paper is to present algorithms that produce adaptive policies of constant robustness ratios for \submaxUC.
Moreover, we consider a special type of adaptive policy called a \emph{universal} policy. 
A universal policy selects items following a precomputed order of items regardless of the observations made while packing.  
Thus, a universal policy is identified with a sequence of the given items. 
This is in contrast to general adaptive policies, where the next item to try can vary with the observations made up to that point. 
We present an algorithm that produces a randomized universal policy that achieves a constant robustness ratio.

If the cancellation is not allowed, then there is no difference between
adaptive and universal policies because the selection terminates once a selected item does
not fit the knapsack.
In this case, we observe that no randomized adaptive policy achieves a constant robustness ratio.
Due to this hardness, we consider a \emph{stochastic} knapsack capacity
when the cancellation is not allowed. 
In this situation, we assume that the knapsack capacity is determined according to some probability distribution and the information of the distribution is available.
Based on this assumption, we compute a sequence of items as a solution. 
When the knapsack capacity is realized, the items in the prefix of the sequence are selected
so that their total size does not exceed the realized capacity.
The objective of the problem is to maximize the expected value of the
submodular function $f$ for the selected items. 
We address this problem as the \emph{submodular maximization problem with a stochastic knapsack capacity} (\submax). 
We say that the \emph{approximation ratio} of a sequence is $\alpha \ (\leq 1)$ if its expected objective value is at least $\alpha$ times the maximum expected value of $f$ for any instance. 
The sequence computed in an $\alpha$-robust policy for \submaxUC achieves an $\alpha$-approximation ratio for \submax. 
However, the opposite does not hold, and an algorithm of a constant approximation ratio may exist for \submax
even though no randomized adaptive policy achieves a constant robustness
ratio for \submaxUC.
Indeed, we present such an algorithm.

\subsection{Related studies}
There are a huge number of studies on the submodular maximization problems (e.g., \cite{Sviridenko04}), but we are aware of no previous work on \submaxUC or \submax.
Regarding studies on the stochastic setting of the problem, several papers proposed concepts of submodularity for random set functions and discussed adaptive policies to maximize those functions~\cite{AsadpourNS08,GolovinK11}.
There are also studies on 
the submodular maximization over an item set in which
each item is activated stochastically~\cite{AdamczykSW16,FeldmanSZ16,GuptaNS17}.
However, as far as we know, there is no study on the problem with stochastic constraints.

When the objective function is modular (i.e., the function returns the sum
of the values associated with the selected items),
the submodular maximization problem with a knapsack constraint is
equivalent to the classic knapsack problem.
For the knapsack problem, there are numerous studies on the stochastic sizes and rewards of items~\cite{DeanGV08,GuptaKMR11,Ma14}.
This problem is called the stochastic knapsack problem. 
Note that this is different from the knapsack problem with a stochastic capacity (\knapsack),
and there is no direct relationship between them.
 However, we observe that most of the algorithms for the stochastic knapsack problem can be applied to \knapsack.
 Indeed, 
one of our algorithms for \submax is based on the idea of Gupta et al.~\cite{GuptaKMR11}
 for the stochastic knapsack problem.

The covering version of \knapsack is
studied in a context of single machine scheduling with nonuniform processing
speed. This problem is known to be strongly NP-hard~\cite{HohnJ15},
which implies that pseudo-polynomial time algorithms are unlikely to exist.
This is in contrast to the fact that the classic knapsack problem and its covering version admit pseudo-polynomial time algorithms.
Megow and Verschae~\cite{MegowV13} gave a PTAS for the covering version of \knapsack.

To the best of our knowledge, \knapsack itself has not been studied well. The only previous study we are aware of is
the thesis of Dabney~\cite{dabney2010ptas}, wherein a PTAS is presented
for the problem.
Since the knapsack problem and its covering version are equivalent in
the existence of exact algorithms, the strongly NP-hardness of the
covering version implies the same hardness for \knapsack.

Regarding the \emph{knapsack problem with an unknown capacity} (\knapsackUC), Megow and Mestre~\cite{MM2013} mentioned that no deterministic policy achieves a constant robustness ratio when cancellation is not allowed. 
They presented an algorithm that constructs for each instance a policy whose robustness ratio is arbitrarily close to the one of an optimal policy that achieves the largest robustness ratio for the instance. 
When cancellation is allowed, Disser et al.~\cite{DisserKMS17}
provided a deterministic $1/2$-robust universal policy for \knapsackUC.
They also
proved that no deterministic adaptive policy achieves a robustness ratio better
than $1/2$, which means that the robustness ratio of their deterministic
universal policy is best possible even for any deterministic adaptive policy.

\subsection{Contributions}
\begin{table}[tb]
\caption{Summary of main results in this paper}
\label{table:contributions}
	\begin{center}
		\begin{tabular}{c|c|c}\hline 
& with cancellation & without cancellation \\ \hline
unknown&  
\multirow{4}{*}{
\begin{tabular}{c}
rand.~adaptive $(1-1/{e})/2$-robust\\
rand.~universal $(1-{1}/{\sqrt[4]{e}})/2$-robust\\
det.~adaptive $2(1-{1}/{e})/21$-robust \\
\end{tabular}
}
		 &
\multirow{2}{*}{
 no constant robustness ratio
}\\
capacity &    &\\ \cline{1-1} \cline{3-3}

stochastic &  & rand.~pseudo-poly time $(1/4-o(1))$-approx.\\
capacity &  & rand. $((1-{1}/{\sqrt[4]{e}})/4-\epsilon)$-approx.
\\ \hline
		\end{tabular}
	\end{center}
\end{table}

Contributions in this paper are summarized in Table~\ref{table:contributions}.
 For the case where cancellation is allowed,
we present three polynomial-time algorithms
for {\submaxUC}.
These algorithms produce

 \begin{itemize}
  \setlength{\itemsep}{0pt}
\item a randomized adaptive policy of robustness ratio $(1-1/e)/2$ (Section~\ref{sec:rand_adaptive});
\item a deterministic adaptive policy of robustness ratio $2(1-1/e)/21$ (Section \ref{sec:det_adaptive}); 
\item a randomized universal policy of robustness ratio $(1-1/\sqrt[4]{e})/2$ (Section \ref{sec:rand_universal}). 
 \end{itemize}

Our algorithms are constructed based on a simple greedy algorithm~\cite{Leskovec+2007} for the monotone submodular maximization problem with a knapsack constraint. 
The greedy algorithm outputs a better solution among two candidates, one of which is constructed greedily based on the increase in the objective function value per unit of size, and the other of which is based on the increase in the objective function value.
In our randomized adaptive policy,
we achieve the robustness ratio $(1-1/e)/2$ by guaranteeing that each of the two candidate solutions is output by our policy with probability $1/2$.

We convert this randomized policy into a deterministic one by mixing the
two strategies which correspond to the two candidate solutions.
We remark that the same approach is taken for \knapsackUC to construct a deterministic $1/2$-robust universal policy by Disser et al.~\cite{DisserKMS17}. 
They call an item \emph{swap item} if it corresponds to a single-item solution for some knapsack capacity.
A key idea in their policy is to pack swap items earlier than the others.
However, their technique fully relies on the property that the objective function is modular, and their choice of swap items is not suitable for \submaxUC. 
In the present paper, we introduce a new notion of \emph{\sv{} items}.
This enables us to design a deterministic $2(1-1/e)/21$-robust policy.
We remark that our proof technique is also different from the standard one used in related work. 
Moreover, we modify the randomized adaptive policy to obtain the randomized universal policy. 
A key idea here is to guess a capacity by a doubling strategy. 

We also show that no randomized adaptive policy achieves a robustness ratio better than $8/9$
 for \knapsackUC  (Section~\ref{sec:hardness_cancel}).
 It is known that the robustness ratio achieved by deterministic policies for this problem is at most $1/2$~\cite{DisserKMS17},
 but there was no upper bound on the robustness ratio for randomized
 policies.
 Disser~et~al.~\cite{DisserKMS17} mentioned that
 it is an interesting open problem to improve the ratio $1/2$
 by randomized policies.
 Our upper bound implies that, even if randomized policies are
 considered, the ratio cannot be improved better than $8/9$.
 In addition, we show that no deterministic policy
 achieves robustness ratio better than $(1+\sqrt{5})/4$
 and no randomized policy achieves robustness ratio better than
 $(5+\sqrt{5})/8$
 for \submaxUC
even when each item has a unit size
(Section~\ref{app.unitcase}).

When cancellation is not allowed, it has been
already known that \knapsackUC
admits no deterministic universal policy of a constant
robustness ratio~\cite{MM2013}. We advance this hardness result by showing that no randomized adaptive policy achieves a constant robustness ratio
(Section~\ref{sec:hardness_nocancel}).

For \submax
without cancellation,
we present the following two algorithms:
\begin{itemize}
 \item a pseudo-polynomial time randomized algorithm of
       approximation ratio $1/4-o(1)$ (Section~\ref{subsec.pseudo-stochastic});

       \item a polynomial-time randomized algorithm of 
       approximation ratio $(1-1/\sqrt[4]{e})/4-\epsilon$ for
       any small constant $\epsilon>0$ (Section \ref{sec:poly_16approx}).
\end{itemize}

The former algorithms is based on the observation that
the problem can be reduced to the submodular maximization problem with an interval independent constraint. 
The latter algorithm is based on the 
idea of Gupta et al.~\cite{GuptaKMR11} for the stochastic knapsack
problem.
Gupta et al.\ regarded the knapsack capacity as a time limit,
and showed that a rounding algorithm for a time-index linear program (LP)
gives an adaptive policy for the stochastic knapsack problem.
Although the formulation size of the time-index LP is not polynomial,
a simple doubling technique reduces the formulation size to polynomial
with a loss of the approximation ratio.
In our algorithm for \submax, 
we first introduce a time-index convex relaxation of the problem using
the multilinear extension of the objective function,
and show that the rounding algorithm of Gupta~et~al.\
is a monotone contention resolution scheme
 for any realization of the knapsack capacity.
 This observation gives
a pseudo-polynomial time algorithm of approximation ratio $(1-1/\sqrt[4]{e})/2-o(1)$ for 
\submax.
We then transform it into a polynomial-time algorithm. 
This transformation requires a careful sketching of knapsack capacity,
which was not necessary for the stochastic knapsack problem.

\subsection{Organization}
The rest of this paper is organized as follows. Section~\ref{sec.prelim}
gives notations and preliminary facts used in this paper. 
Section~\ref{sec:submaxuc}
presents the adaptive policies for \submaxUC with cancellation.
Section~\ref{sec:upperbound}
provides the upper-bounds on the robustness ratios.
Section~\ref{sec.stochastic} presents the approximation algorithms 
for \submax without cancellation.
Section~\ref{sec:conclusion} concludes the paper. 

%% file: preliminaries.tex
\section{Preliminaries}
\label{sec.prelim}

 In this section, we define terminologies used in this paper, and
 introduce existing results which we will use.

\para{Maximization of monotone submodular functions}
The inputs of the problem are a set $I$ of  $n$ items
and a nonnegative set function $f:2^I\to\mathbb{R}_+$.
In this paper, we assume that (i) $f$ satisfies $f(\emptyset)=0$,
(ii) $f$ is \emph{submodular} (i.e., $f(X)+f(Y)\ge f(X\cup Y)+f(X\cap Y)$ for any $X,Y
\subseteq I$),
and (iii) $f$ is \emph{monotone} (i.e., $f(X)\le f(Y)$ for any $X,Y\subseteq I$ with $X
\subseteq Y$).
The function $f$ is given as an oracle that returns the value of $f(X)$ for any query $X
\subseteq I$.
Let $\mathcal{I} \subseteq 2^I$ be any family 
such that $X \subseteq Y \in \mathcal{I}$
implies $X \in \mathcal{I}$. The $\mathcal{I}$-constrained
submodular maximization problem
seeks $X \in \mathcal{I}$
that maximizes $f(X)$.

We focus on the case where $\mathcal{I}$ corresponds to a knapsack constraint.
Namely, each item $i\in I$ is associated with a size $s(i)$,
and $\mathcal{I}$ is defined as $\{X \subseteq I \colon \sum_{i \in X}s(i)\leq C\}$
for some knapsack capacity $C > 0$.
We assume that the item size $s(i)$ ($i \in I$)
and the knapsack capacity $C$ are positive integers.
We denote $\sum_{i \in X}s(i)$ by $s(X)$ for any $X \subseteq I$.

\para{Problem SMPUC}
In \submaxUC, the knapsack capacity $C$ is unknown.
We can see whether an item set fits the knapsack only when an item is added
to the item set.

A solution for \submaxUC is an adaptive policy $\mathcal{P}$, which is represented as a binary decision tree that contains every item at most once along each path from the root to a leaf. 
Each node of the decision tree is an item to try packing into the knapsack. 
A \emph{randomized} policy is a probability distribution over binary decision trees. 
One of the decision trees is selected according to the probability distribution. 
For a fixed capacity $C$, the output of a policy $\mathcal{P}$ is an item set denoted by $\mathcal{P}(C)\subseteq I$ obtained as follows.
We start with $\mathcal{P}(C) = \emptyset$ and check whether the item $r$ at the root of $\mathcal{P}$ fits the knapsack, i.e., 
whether $s(r) + s(\mathcal{P}(C)) \le C$. 
If the item fits, then we add $r$ to $\mathcal{P}(C)$ and continue packing recursively with the left subtree of $r$. 
Otherwise, we have two options:
when cancellation is allowed,
we discard $r$ and continue packing recursively with the right subtree of $r$;
when cancellation is not allowed,
we discard $r$ and output $\mathcal{P}(C)$ to terminate the process.

When a policy does not depend on the observation made while packing, we call such a policy \emph{universal}. 
Since every path from the root to a leaf in a universal policy is identical, we can identify a universal policy with a sequence $\Pi = (\Pi_1, \ldots, \Pi_n)$ of items in $I$. 
For a fixed capacity $C$, the output of a universal policy, denoted by $\Pi(C)$, is constructed as follows.
We start with $\Pi(C)=\emptyset$, and add items to $\pi(C)$ in $n$ iterations.
In the $i$th iteration, we check whether $s(\Pi(C)) + s(\Pi_i) \leq C$ holds or not.
If true, then $\Pi_i$ is added to $X$.
Otherwise, $\Pi_i$ is discarded, and we proceed to the next iteration when cancellation is allowed,
and we terminate the process when cancellation is not allowed.
 
\para{Problem SMPSC}
In \submax, the knapsack capacity $C$ is given according to some probability distribution.
Let $T=\sum_{i \in I}s(i)$.
For each $t \in [T]:=\{0,1,\ldots,T\}$,
we denote by $p(t)$ the probability that the knapsack capacity is $t$.
We assume that the probability is given to an algorithm through an oracle that returns the value of $\sum_{t' =t}^T p(t')$ for any query $t \in [T]$.
Hence, the input size of a problem instance is $\Theta(n\log T)$
and an algorithm runs in pseudo-polynomial time if its running time depends on $T$ linearly.
A solution for \submax is a universal policy, i.e., a sequence $\Pi=(\Pi_1,\ldots,\Pi_n)$ of the items in $I$. 
When a capacity $C$ is decided, the output $\Pi(C)$ of $\Pi$ is constructed in the same way as universal policies for \submaxUC. 
The objective of \submax is to find a sequence $\Pi$ that maximizes $\mathbb{E}_C[f(\Pi(C))]$.

\para{Multilinear extension, continuous greedy, and contention
    resolution scheme}
From any vector $x \in [0,1]^I$,
we define a random subset $R_x$ of $I$
so that each $i \in I$ is included in $R_x$ with probability $x_i$,
where the inclusion of $i$ is independent from the inclusion of the other items.
For a submodular function $f\colon 2^I \rightarrow \Rset_+$, 
its \emph{multilinear extension} $F\colon [0,1]^I \rightarrow \Rset_+$
is defined by $F(x)=\mathbb{E}[f(R_x)]=\sum_{X \subseteq I} f(X)\prod_{i \in X}x_i \prod_{i' \in I\setminus X}(1-x_{i'})$ for all $x \in [0,1]^I$.
This function $F$
satisfies the \emph{smooth monotone submodularity},
that is,
$\partial F(x)/\partial x_i \geq 0$ for any $i \in I$
and $\partial^2 F(x)/(\partial x_i \partial x_j) \leq 0$ for any $i,j \in I$.
Although it is hard to compute the value of $F(x)$ exactly,
it usually can be approximated with arbitrary precision by a Monte-Carlo
simulation.
In this paper, to make discussion simple, we assume that $F$ can be
    computed exactly.

A popular approach for solving the $\mathcal{I}$-constrained submodular
maximization problem is to use a continuous relaxation of the problem.
Let $P \subseteq [0,1]^I$ be a polytope
in which each integer vector in $P$ is the incidence vector of a member of $\mathcal{I}$.
Then, $\max_{x \in P}F(x) \geq \max_{X \in \mathcal{I}}f(X)$ holds.
In this approach, it is usually assumed that $P$ is \emph{downward-closed}
(i.e., if $x,y \in [0,1]^I$ satisfies $y\leq x \in P$, then $y \in P$),
and \emph{solvable}
(i.e., the maximization problem $\max_{x \in P} \sum_{i\in I} w_i x_i$
can be solved in polynomial time
for any $w \in \Rset_+^I$).

Calinescu et al.~\cite{CalinescuCPV11}
gave an algorithm called \emph{continuous greedy}
for a continuous maximization problem $\max_{x \in P} F(x)$
over a solvable downward-closed polytope $P$.
They proved that the continuous greedy
outputs a vector $x \in P$
such that $F(x) \geq (1-1/e-o(1)) \max_{x'\in P}F(x')$.
Feldman~\cite{Feldmanphd}
extended its analysis by observing that
the continuous greedy algorithm with stopping time $b \geq 0$
outputs a vector $x \in [0,1]^I$ such that $x/b \in P$
and $F(x) \geq (1-e^{-b}-o(1)) \max_{X \in \mathcal{I}}f(X)$
(The performance guarantee depending on the stopping time
is originally given 
for the measured continuous greedy algorithm
proposed by \cite{FeldmanNS11}).
It is easy to see that his analysis can be modified to prove a slightly stronger
result $F(x) \geq (1-e^{-b}-o(1)) \max_{x' \in P}F(x')$.
In addition, this analysis requires
only the smooth monotone submodularity as a property of $F$.

A fractional vector $x \in P$ can be rounded into an integer vector
by a contention resolution scheme.
Let $b,c \in [0,1]$. 
For a vector $x$, we denote $\supp(x)=\{i \in I \colon x_i > 0\}$.
We consider an algorithm
that receives $x \in \{z \colon z/b \in P \}$ and $A \subseteq I$ as inputs
and returns a random subset $\pi_x(A) \subseteq A\cap \supp(x)$.
Such an algorithm is called 
 \emph{$(b,c)$-balanced contention resolution
 scheme} if
 $\pi_x(A) \in \mathcal{I}$
 with probability 1
 for all $x$ and $A$,
 and 
 $\Pr[i \in \pi_x(R_x) \mid i \in R_x] \geq c$
 holds for all $x$ and $i \in \supp(x)$
 (recall that $R_x$ is the random subset of $I$ determined from $x$).
 It is also called \emph{monotone} if
 $\Pr[i \in \pi_x(A)] \geq \Pr[i \in \pi_x(A')]$
 for any $i \in A \subseteq A' \subseteq I$.
 If a monotone $(b,c)$-balanced contention resolution scheme is
 available,
 then we can achieve the approximation ratio claimed in the following
 theorem
 by applying it to a fractional vector computed by the measured
 continuous greedy algorithm with stopping time $b$.
 This fact is summarized as in the following theorem.
 
  \begin{theorem}[\cite{FeldmanNS11}]
  \label{thm.contscheme}
 If there exists a monotone $(b,c)$-balanced contention resolution
 scheme for $\mathcal{I}$, then the
 $\mathcal{I}$-constrained submodular maximization problem 
 admits an approximation algorithm 
of ratio $(1-e^{-b})c-o(1)$ for any nonnegative monotone submodular function.
  \end{theorem}

%% file: universal.tex
\section{Adaptive policies for \submaxUC with cancellation}
\label{sec:submaxuc}

\subsection{Randomized $(1-1/e)/2$-robust policy}
\label{sec:rand_adaptive}

In this subsection, we present a randomized adaptive policy for \submaxUC in the situation that the cancellation is allowed. 
The idea of our algorithm is based on a simple greedy algorithm~\cite{Leskovec+2007} for the submodular maximization problem with a knapsack constraint. 
The greedy algorithm generates two candidate item sets.  
One set is obtained greedily by repeatedly inserting  an item maximizing the increase in the objective function value per unit of size. 
The other is obtained similarly by packing an item maximizing the increase in the objective function value. 
Then, the algorithm returns the set with the larger value, which leads to a $(1-1/e)/2$-approximation solution.

The idea of choosing a better solution is not suitable for \submaxUC, in which we cannot remove items from the knapsack. 
We resolve this issue by generating at random one of two policies $\mathcal{P}^1$ and $\mathcal{P}^2$ that are analogous to the above two greedy methods. 
One policy $\mathcal{P}^1$ corresponds to the greedy algorithm based on the increase in the objective function value per unit of size. 
We formally present this policy as Algorithm \ref{alg:greedy1}. 
The item $i_j$ corresponds to a node of depth $j-1$ in $\mathcal{P}^1$. 
We remark that Algorithm \ref{alg:greedy1} chooses $i_j$ independently of the knapsack capacity, but the choice depends on the observations which items fit the knapsack and which items did not so far. 
For generality of the algorithm, we assume that the algorithm receives 
an initial state $U$ of the knapsack, which is defined as a subset of $I$
(this will be used in Section~\ref{sec:det_adaptive}).

\begin{algorithm}[t]
\caption{Greedy algorithm $\mathcal{P}^1$ for $(I, U)$}\label{alg:greedy1}
$X\ot U$, $R\ot I \setminus U$\;
\ForEach{$j=1,\ldots, |I\setminus U|$}{
  let $i_j\in\argmax\left\{(f(X\cup\{i\})-f(X))/s(i) \colon i\in R\right\}$\;
  \If(\tcp*[h]{left subtree}){$i_j$ fits the knapsack (i.e., $s(X)+s(i_j)\le C$)}{$X\ot X\cup\{i_j\}$}
  \Else(\tcp*[h]{right subtree}){discard $i_j$}
 $R\ot R\setminus\{i_j\}$\;
 }
    \Return{$X$}\;
\end{algorithm}

The policy $\mathcal{P}^2$ tries to pack items greedily based on the increase in the objective function value. 
Our algorithm is summarized in Algorithm \ref{alg:submodular_universal}. 
We remark that Algorithm \ref{alg:submodular_universal} chooses the item $i_j$ for each iteration $j$ in polynomial time with respect to the cardinality of $I$. 

\begin{algorithm}[t]
\caption{Randomized $(1-1/e)/2$-robust adaptive policy}\label{alg:submodular_universal}
flip a coin\;
\lIf(\tcp*[h]{policy $\mathcal{P}^1$}){head}{execute Algorithm \ref{alg:greedy1} for $(I, \emptyset)$}
\Else(\tcp*[h]{policy $\mathcal{P}^2$}){
  $X\ot \emptyset$, $R\ot I$\;
  \ForEach{$j=1,\ldots, |I|$}{
    let $i_j\in\argmax\left\{f(X\cup\{i\})-f(X)\colon i\in R\right\}$\;
    \If(\tcp*[h]{left subtree}){$i_j$ fits the knapsack (i.e., $s(X)+s(i_j)\le C$)}{$X\ot X\cup\{i_j\}$}
    \Else(\tcp*[h]{right subtree}){discard $i_j$}
    $R\ot R\setminus\{i_j\}$\;
  }
    \Return{$X$}\;
 }
\end{algorithm}

We analyze the robustness ratio of Algorithm \ref{alg:submodular_universal}. 
In the execution of Algorithm \ref{alg:greedy1} for $(I, U)$ under some capacity $C$, we call the order $(i_1, \ldots, i_{|I \setminus U|})$ of items in $I \setminus U$ the \emph{greedy order} for $(I, U)$ with capacity $C$, where $i_j$ is the $j$th selected item at line 3. 
A key concept in the analysis of Algorithm \ref{alg:submodular_universal} is to focus on the first item in the greedy order that is a member of $\OPT_C$ but is spilled from $\mathcal{P}^1(C)$. 
The following lemma is useful in analysis of Algorithm \ref{alg:submodular_universal} and also algorithms given in subsequent sections. 

\begin{lemma}\label{lemma:submodular_greedy}
Let $C, C'$ be any positive numbers, and let $q$ be the smallest index such that $i_q \in \OPT_{C'}$ and $i_q \not\in \mathcal{P}^1(C)$ (let $q=\infty$ if there is no such index). 
When Algorithm \ref{alg:greedy1} is executed for $(I, U)$ with capacity $C$, 
it holds for any index $j$ that
  \begin{align*}
  & f(((\mathcal{P}^1(C) \cup \OPT_{C'}) \cap \{i_1, \ldots, i_{j}\}) \cup  U)
    \ge \left(1-\exp\left(-\frac{s((\mathcal{P}^1(C) \cup \OPT_{C'}) \cap \{i_1, \ldots, i_{j}\})}{C'}\right)\right)\cdot f(\OPT_{C'}). 
  \end{align*}
  Moreover, $(\mathcal{P}^1(C) \cup \OPT_{C'}) \cap \{i_1, \ldots, i_{j}\}=
  \mathcal{P}^1(C) \cap \{i_1, \ldots, i_{j}\}$ holds for any $j < q$. 
\end{lemma}
To prove Lemma \ref{lemma:submodular_greedy}, we show the following two lemmas. 

\label{app.greedy}
\begin{lemma}\label{lemma:submodular_decompose}
  For any $X\subseteq Y\subseteq I$, we have
  \begin{align*}
    f(Y)\le f(X)+\sum_{i\in Y\setminus X}(f(X\cup\{i\})-f(X)).
  \end{align*}
\end{lemma}
\begin{proof}
  Suppose that $Y\setminus X=\{a_1,\dots,a_l\}$, where $l=|Y\setminus X|$. 
  Let $X_j=X\cup\{a_1,\dots,a_j\}$ for each $j=1, \ldots, l$.
  Note that $X_0=X$ and $X_l=Y$.
  Since $f$ is submodular, we have $f(X_j)-f(X_{j-1})\le f(X\cup\{a_j\})-f(X)$. 
  By summing both sides over $j$, we obtain
  $f(Y)-f(X)=\sum_{j=1}^l(f(X_j)-f(X_{j-1}))\le \sum_{j=1}^l(f(X\cup\{a_j\})-f(X))$,  
  which proves the lemma.
\end{proof}

\begin{lemma}\label{lemma:submodular_inc}
Let $C$ be any positive number. 
  For any $X\subseteq Y \subseteq I$ and any $i \in I\setminus Y$ such that 
$\OPT_C \setminus X \subseteq I \setminus Y$ and 
\begin{align*}
\frac{f(X\cup\{i\})-f(X)}{s(i)}
=\max \left\{\frac{f(X\cup\{v\})-f(X)}{s(v)}\colon v \in I \setminus Y \right\},
\end{align*}
it holds that
  \begin{align*}
    f(X\cup\{i\})-f(X)\ge \frac{s(i)}{C}(f(\OPT_C)-f(X)).
  \end{align*}
\end{lemma}
\begin{proof}
By applying Lemma \ref{lemma:submodular_decompose} with $Y=\OPT_C \cup X \ (\supseteq X)$, we have $f(\OPT_C \cup X) \leq f(X)+\sum_{v \in \OPT_C\setminus X}(f(X\cup\{v\})-f(X))$. 
In addition, $f(\OPT_C) \leq f(\OPT_C \cup X)$ holds by monotonicity. 
Thus, we have
  \begin{align*}
    f(\OPT_C)-f(X)
    &\le \sum_{v\in \OPT_C\setminus X}(f(X\cup\{v\})-f(X))
    = \sum_{v\in \OPT_C\setminus X}s(v)\cdot \frac{f(X\cup\{v\})-f(X)}{s(v)}\\
    &\le \left(\sum_{v\in \OPT_C\setminus X}s(v) \right)\cdot \frac{f(X\cup\{i\})-f(X)}{s(i)}
    \le C\cdot \frac{f(X\cup\{i\})-f(X)}{s(i)}.
  \end{align*}
\end{proof}

\begin{proof}[Proof of Lemma \ref{lemma:submodular_greedy}]
We denote $Y_j = U \cup \{i_1, \ldots, i_{j}\}$ and $Z_j = (\mathcal{P}^1(C) \cup \OPT_{C'}) \cap \{i_1, \ldots, i_j\}$ for each $j$. 
Let $Y_0 = U$ and $Z_0 = \emptyset$. 
  We first prove by induction on $j$ that $f(Z_{j} \cup U)\ge (1-\exp(-s(Z_{j})/{C'}))\cdot f(\OPT_{C'})$ for any $j$.
  For $j=0$, the statement clearly holds because $1-\exp(-s(Z_0)/{C'})=0$ and $f(Z_0 \cup U)\geq f(\emptyset) = 0$.

Suppose that $f(Z_{j-1} \cup U)\ge (1-\exp(-s(Z_{j-1})/{C'}))\cdot f(\OPT_{C'})$ holds for some $j  \geq 0$.
If $i_j \not\in \mathcal{P}^1(C) \cup \OPT_{C'}$, then we have $Z_j = Z_{j-1}$, and hence it is easy to see that $f(Z_{j} \cup U) \geq (1-\exp(-s(Z_{j})/{C'}))\cdot f(\OPT_{C'})$. 
In the following, we may assume that $i_j \in \mathcal{P}^1(C) \cup \OPT_{C'}$. 
Note that $Z_{j-1} \cup \{i_j\} = Z_j$. 
We observe that by choice of $i_j$ in Algorithm \ref{alg:greedy1}, it holds that
\begin{align*}
\frac{f((Z_{j-1} \cup U)\cup\{i_j\})-f(Z_{j-1} \cup U)}{s(i_j)} 
=\max \left\{\frac{f((Z_{j-1} \cup U)\cup\{v\})-f(Z_{j-1} \cup U)}{s(v)}
\colon v\in I\setminus Y_{j-1}\right\}. 
\end{align*}
Because $\OPT_{C'} \setminus (Z_{j-1} \cup U) \subseteq I \setminus Y_{j-1}$, we can apply Lemma \ref{lemma:submodular_inc} with $X=Z_{j-1} \cup U$, $Y=Y_{j-1}$ and $i=i_j$ to derive 
\begin{align*}
f(Z_{j} \cup U) - f(Z_{j-1} \cup U) \geq \frac{s(i_j)}{C'}(f(\OPT_{C'})-f(Z_{j-1} \cup U)).
\end{align*}
Therefore, by using this inequality, we see that
  \begin{align*}
    f(Z_j \cup U)
    &\ge f(Z_{j-1} \cup U)+\frac{s(i_j)}{C'}(f(\OPT_{C'})-f(Z_{j-1} \cup U))\\
    &=\left(1-\frac{s(i_j)}{C'}\right)f(Z_{j-1} \cup U)+\frac{s(i_j)}{C'}f(\OPT_{C'})\\
    &\ge \left(1-\frac{s(i_j)}{C'}\right)\cdot\left(1-\exp\left(-\frac{s(Z_{j-1})}{C'}\right)\right)\cdot f(\OPT_{C'})+\frac{s(i_j)}{C'}f(\OPT_{C'})\\
    &= \left(1-\left(1-\frac{s(i_j)}{C'}\right)\cdot\exp\left(-\frac{s(Z_{j-1})}{C'}\right)\right)\cdot f(\OPT_{C'})\\
    &\ge \left(1-\exp\left(-\frac{s(i_j)}{C'}\right)\cdot\exp\left(-\frac{s(Z_{j-1})}{C'}\right)\right)\cdot f(\OPT_{C'})
    = \left(1-\exp\left(-\frac{s(Z_{j})}{C'}\right)\right)\cdot f(\OPT_{C'}),
  \end{align*}
  where the last inequality holds because $1-x\le\exp(-x)$ for any $x$.
  
It remains to show that $Z_j = \mathcal{P}^1(C) \cap \{i_1, \ldots, i_{j}\}$ for any $j < q$. 
The inclusion $\supseteq$ is clear. 
For any $j \leq q$, we have $\OPT_{C'} \cap \{i_1, \ldots, i_j\} \subseteq \mathcal{P}^1(C) \cap \{i_1, \ldots, i_j\}$ by choice of $q$. 
Thus, $Z_j = (\OPT_{C'} \cap \{i_1, \ldots, i_j\} ) \cup (\mathcal{P}^1(C) \cap \{i_1, \ldots, i_{j}\}) \subseteq \mathcal{P}^1(C) \cap \{i_1, \ldots, i_{j}\}$. 
This completes the proof.
\end{proof}

We are ready to analyze the robustness ratio of Algorithm \ref{alg:submodular_universal}. 
For any $C \in \mathbb{R}_+$, we denote $I_C = \{ i \colon s(i) \leq C\}$ and denote by $i^C$ an item with the largest value in this set, i.e., $i^C \in \argmax\{f(\{i\})\colon i\in I_C\}$. 

\begin{theorem}\label{thm:unknown_rand_adaptive}
  Algorithm \ref{alg:submodular_universal} is a randomized $(1-1/{e})/2 > 0.316$-robust adaptive policy.
\end{theorem}
\begin{proof}
Let $\mathcal{P}^1$ (respectively, $\mathcal{P}^2$) be the adaptive policy in Algorithm \ref{alg:submodular_universal} when the coin comes up head (respectively, tail). 
Suppose that the given capacity is $C$.
The expected value of the output by Algorithm \ref{alg:submodular_universal} is $\ALG_C = (f(\mathcal{P}^1(C)) + f(\mathcal{P}^2(C)))/2$. 
If $\OPT_C\subseteq \mathcal{P}^1(C)$, then we get $\ALG_C\ge f(\mathcal{P}^1(C))/2\ge f(\OPT_C)/2$.
Assume that $\OPT_C\not\subseteq \mathcal{P}^1(C)$.
Let $(i_1,i_2,\dots,i_n)$ be the greedy order for $(I, \emptyset)$ with capacity $C$. 
Let $q$ be the smallest index such that $i_q \in \OPT_C$ and $i_q \not\in \mathcal{P}^1(C)$. 
We have $f(\mathcal{P}^2(C))\ge f(\{i^C\})\ge f(\{i_q\})$ because $i^C \in \mathcal{P}^2(C)$ and $i_q \in I_C$.
By the monotonicity and submodularity of $f$, we have
\begin{align*}
  \ALG_C&
  \ge \frac{f(\mathcal{P}^1(C) \cap \{i_1,\dots,i_{q-1}\})+f(\{i_q\})}{2} 
  \ge \frac{f((\mathcal{P}^1(C) \cap \{i_1,\dots,i_{q-1}\}) \cup \{i_q\})}{2}.
\end{align*}
Note that by the choice of $q$, we have $(\mathcal{P}^1(C) \cap \{i_1, \ldots, i_{q-1}\}) \cup \{i_q\} = (\mathcal{P}^1(C) \cup \OPT_{C}) \cap \{i_1, \ldots, i_{q}\}$ and $s((\mathcal{P}^1(C) \cap \{i_1, \ldots, i_{q-1}\}) \cup \{i_q\}) > C$.
Thus, Lemma \ref{lemma:submodular_greedy} implies that ${f((\mathcal{P}^1(C) \cup \OPT_C) \cap \{i_1, \ldots, i_{q}\})} \ge {(1-1/e)}f(\OPT_C)$, and hence we have $\ALG_C \geq {(1-1/e)}f(\OPT_C)/2$. 
\end{proof}

\subsection{Deterministic $2(1-1/e)/21$-robust policy}
\label{sec:det_adaptive}

In this subsection, we present a deterministic adaptive policy for \submaxUC by modifying Algorithm \ref{alg:submodular_universal}. 
To this end, let us review the result of Disser et al.~\cite{DisserKMS17} for \knapsackUC, which is identical to \submaxUC with modular objective functions. 
They obtained a deterministic $1/2$-robust universal policy for \knapsackUC based on the greedy order for $(I, \emptyset)$ with the given capacity. 
Their policy first tries to insert items with large values, which are called \emph{swap items}. 
For a greedy order $(i_1, \ldots, i_n)$, an item $i_j$ is called a swap item if $f(\{i_j\})\ge f(\mathcal{P}^1(C))$ for some capacity $C$ such that $i_j$ is the first item that overflows the knapsack when the items are packed in the greedy order. 
The key property is that the greedy order does not depend on the capacity $C$ when $f$ is modular. 
This enables them to determine swap items from the unique greedy order, and to obtain a deterministic universal policy. 

On the other hand, it is hard to apply their idea to our purpose. 
The difficulty is that the greedy order varies according to the capacity when $f$ is submodular. 
Thus, the notion of swap items is not suitable in \submaxUC for choosing the items that should be tried first. 
In this paper, we introduce \emph{\sv{} items}, which are items $i$ satisfying $f(\{i\}) \geq 2\cdot f(\OPT_{s(i)/2})~(=2\cdot \max\{f(X) \colon s(X) \leq s(i)/2\})$. 
In the design of our algorithm, we use a polynomial-time $\app$-approximation algorithm that computes $f(\OPT_{s(i)/2})$;
for example, $\app =1-1/e$~\cite{Sviridenko04}.
Our algorithm calculates the set $\SV$ of the \sv{} items in a sense of $\app$-approximation.
To be precise, it holds that $f(\{i\})\ge 2\app\cdot f(\OPT_{s(i)/2})$ for any item $i\in \SV$ and $f(\{i\})\le 2\cdot f(\OPT_{s(i)/2})$ for any item $i \not\in S$.

Our algorithm first tries to insert items in $\SV$ until one of these items fits the knapsack (or all the items have been canceled) and then it executes Algorithm \ref{alg:greedy1} with the remaining items.
We remark that, unlike the algorithm for \knapsackUC~\cite{DisserKMS17}, 
our algorithm executes Algorithm \ref{alg:greedy1} once a \sv{} item fits the knapsack.
We summarize our algorithm in Algorithm \ref{alg:submodular_det}.

\begin{algorithm}
\caption{Deterministic $2\app/21$-robust policy}\label{alg:submodular_det}
$U \ot \emptyset$, $R\ot I$, $\SV\ot \emptyset$\;
\ForEach{$i\in I$}{
  Let $L$ be a $\app$-approximate solution to $\max\{f(X)\colon s(X)\le s(i)/2,~X\subseteq I\}$\;
  \lIf{$f(\{i\})\ge 2 f(L)$}{$\SV\ot \SV\cup\{i\}$}
}
\While{$U=\emptyset$ and $\SV\cap R\ne\emptyset$}{
  let $i\in\argmax\{f(\{i\})\colon i\in \SV\cap R\}$\;\label{alg:sd_max}
  \lIf(\tcp*[h]{left subtree}){$i$ fits the knapsack (i.e., $s(i)\le C$)}{$U\ot \{i\}$}
  \lElse(\tcp*[h]{right subtree}){discard $i$}
  $R\ot R\setminus\{i\}$\;
}
execute Algorithm \ref{alg:greedy1} for $(R\cup U, U)$\;\label{alg:sd_greedy}
\end{algorithm}

Note that our algorithm constructs $\SV$ and decides which item to try in polynomial time. 
In the rest of this subsection, we let $R$, $U$, and $\SV$ denote the item sets at the beginning of line~\ref{alg:sd_greedy}.
Then $U$ is empty if $\SV\cap I_C=\emptyset$, and $U$ consists of exactly one item $i^*\in\argmax\{f(\{i\})\colon i\in I_C\cap \SV\}$ otherwise.
The following theorem is the main result of this section. 

\begin{theorem}\label{thm:det_univ_policy}
Algorithm \ref{alg:submodular_det} using a $\app$-approximation algorithm in line 3 is a $\min\{2\app/21,\allowbreak (1-1/\sqrt[3]{e})/3\}$-robust universal policy. 
In particular, it is $2(1-1/e)/21 > 0.060$-robust when $\app=1-1/e$, and it is $(1-1/\sqrt[3]{e})/3 > 0.094$-robust when $\app=1$.
\end{theorem}

We describe the proof idea. 
We discuss the behavior of Algorithm~\ref{alg:submodular_det}
when the given capacity is $C$.
Since the output of Algorithm \ref{alg:submodular_det} is $\mathcal{P}^1(C)$ for $(R\cup U, U)$, one can think of a similar proof to the one of Theorem \ref{thm:unknown_rand_adaptive}. 
However, we may not be able to use the value $f(\{i_q\})$ in the evaluation of $f(\mathcal{P}^1(C))$ here because $\mathcal{P}^1(C)$ may not contain any items that bound $f(\{i_q\})$.
We show the theorem using a different approach. 
A basic idea is to divide $\OPT_C$ into several subsets $A_1, \ldots, A_k$ and derive a bound $f(\OPT_C) \leq f(A_1) + \cdots + f(A_k)$ by the submodularity of $f$. 
A key idea is to evaluate each $f(A_i)$ using properties of \sv{} items, which are shown as the following two lemmas. 

\begin{lemma}\label{lemma:det}
It holds that $f(\{i\}) \leq  \max \{ f(U), 2f(\OPT_{s(i)/2})\}$ for any item $i \in I_C$.
\end{lemma}
\begin{proof}
The lemma follows because $f(\{i\}) \leq f(U)$ if $i \in \SV$ and $f(\{i\}) \leq 2 f(\OPT_{s(i)/2})$ if $i \not\in \SV$. 
\end{proof}

\begin{lemma}\label{lemma:det_univ_policy2}
Let $s^* = s(U)$. 
If $s^* \leq C/2$,
then, for any number $x \in [s^*, C/2]$, 
it holds that $f(\OPT_{2x})\le 3  f(\OPT_{x})$. 
\end{lemma}
\begin{proof}
First, suppose that $\OPT_{2x}$ contains some item $i$ with $s(i) > x$ and $i\in \SV$. 
Then we have $f(\OPT_{2x}) \le f(\{i\})+f(\OPT_{2x}\setminus\{i\})$ by the submodularity of $f$.
Since the existence of item $i$ implies $I_C \cap S\neq \emptyset$,
$U$ consists of exactly one item from $\SV$. We denote this item by $i^*$.
  Any item $i' \in \SV$ with $f(\{i'\}) > f(\{i^*\})$ does not fit the knapsack with capacity $C \geq 2x$. 
  This implies that $f(\{i\})\le f(\{i^*\})$. 
 Moreover, we have $f(\{i^*\}) \leq f(\OPT_{x})$ because $s(i^*) = s^* \leq x$. 
 Hence, $f(\{i\}) \leq f(\OPT_x)$ follows.
On the other hand, $f(\OPT_{2x}\setminus\{i\}) \leq f(\OPT_{2x-s(i)})\leq f(\OPT_x)$ holds, where the last inequality follows from $s(i)>x$.
Therefore, we see that
\begin{align*}
f(\OPT_{2x})\leq f(\{i\})+f(\OPT_{2x}\setminus\{i\})\le  2 f(\OPT_{x}).
\end{align*}
  
Second, suppose that $\OPT_{2x}$ contains some item $i$ with $s(i) > x$ and $i\not\in \SV$. 
  Recall that $i \not\in \SV$ implies $f(\{i\} ) \leq 2 f(\OPT_{s(i)/2})$. 
  Since $s(i)/2 \leq x < s(i)$, we see that 
  \begin{align*}
  f(\OPT_{2x}) \le f(\{i\})+f(\OPT_{2x}\setminus\{i\}) \le 2 f(\OPT_{s(i)/2})+f(\OPT_{2x-s(i)}) \le 3 f(\OPT_{x}).
  \end{align*}
  
Finally, assume that all items in $\OPT_{2x}$ have size at most $x$. 
  We can divide $\OPT_{2x}$ into three sets $A_1, A_2, A_3$ with $s(A_j) \leq x \ (j=1,2,3)$. 
  Therefore, we have 
  \begin{align*}
  f(\OPT_{2x}) \le f(A_1) +f(A_2)+f(A_3)  \le 3 f(\OPT_{x}).
  \end{align*}
The lemma follows from the arguments on the three cases. 
\end{proof}

\begin{proof}[Proof of Theorem \ref{thm:det_univ_policy}]
Let $\mathcal{P}$ be the deterministic policy described as Algorithm \ref{alg:submodular_det}.
Suppose that the given capacity is $C$. 
We remark that $\mathcal{P}(C) = \mathcal{P}^1(C)$ for $(R\cup U, U)$. 
We may assume that $\OPT_C\not\subseteq \mathcal{P}(C)$ since otherwise $f(\mathcal{P}(C))=f(\OPT_C)$.
Let $s^*=s(U)$. 
We branch the analysis into two cases: (a) $s^*<C/3$ and (b) $s^*\ge C/3$. 

\para{Case (a)}
We claim that 
\begin{align*}
f(\mathcal{P}(C))\ge (1-1/\sqrt[3]{e})\cdot f(\OPT_{C})/3.
\end{align*}
Since $s^* < C/3<C/2$, Lemma \ref{lemma:det_univ_policy2} indicates that $f(\OPT_{C/2}) \geq f(\OPT_C)/3$. 
We evaluate $f(\OPT_{C/2})$ by using Lemma \ref{lemma:submodular_greedy} with $C'=C/2$. 
We may assume that $\OPT_{C/2}\not\subseteq\mathcal{P}(C)$; otherwise $f(\mathcal{P}(C)) \geq f(\OPT_{C/2})$ holds, which implies that $f(\mathcal{P}(C)) \geq f(\OPT_C)/3$. 
Let $(i_1,\dots,i_{|R|})$ be the greedy order for $(R\cup U,U)$ with capacity $C$.
Let $q'$ be the smallest index such that $i_{q'} \in \OPT_{C/2}$ and $i_{q'}\not\in \mathcal{P}^1(C)$. 
We denote $Z=\mathcal{P}^1(C)\cap\{i_1,\dots,i_{q'-1}\}$. 
As $s(i_{q'})\le C/2$, we see that $s(Z)>C-s^*-s(i_{q'}) \geq C/6$.
Then Lemma \ref{lemma:submodular_greedy} implies that 
\[f(\mathcal{P}(C))\ge f(Z \cup U)\ge (1-1/\sqrt[3]{e})\cdot f(\OPT_{C/2}),\]
and hence the claim follows.

\para{Case (b)}
In this case, we prove the following two claims. 
Note that $U$ is nonempty since $s^* > 0$. 
Let $i^*$ denote the unique item in $U$.

\begin{claim}
$f(\OPT_C) \le 7f(\OPT_{s^*})$. 
\end{claim}
\begin{proof}
Let $T'=\{i\in \OPT_{C}\colon s(i)> s^*\}$. 
Since $C\leq 3s^*$, we observe that $T'$ has at most two items. 
We evaluate $f(\OPT_C)$ by splitting $\OPT_C$ depending on $T'$.

Suppose that $T' = \{i\}$ and $2s^* \leq s(i) \ (\leq C)$. 
It follows that $f(\OPT_C) \le f(\{i\})+f(\OPT_{C}\setminus\{i\})$. 
By Lemma \ref{lemma:det}, we have $f(\{i\}) \leq \max\{f(U), 2f(\OPT_{s(i)/2}) \} \leq \max \{ f(\OPT_{s^*}), 2f(\OPT_{s(i)/2})\}$. 
Since $s(i)/2 \leq C/2 < 2s^*$, it follows that $\max \{ f(\OPT_{s^*}), 2f(\OPT_{s(i)/2})\} \leq 2f(\OPT_{2s^*})$, which is at most $6f(\OPT_{s^*})$ by Lemma \ref{lemma:det_univ_policy2} together with $s^* \leq s(i)/2 \leq C/2$. 
Moreover, we have $f(\OPT_{C}\setminus\{i\}) \leq f(\OPT_{C-s(i)}) \leq f(\OPT_{s^*})$
since $C-s(i)\le 3s^*-2s^*= s^*$. 
Thus, we obtain 
		\begin{align*}
		f(\OPT_C) \le f(\{i\})+f(\OPT_{C}\setminus\{i\}) \le 6 f(\OPT_{s^*})+f(\OPT_{s^*}) 
		= 7f(\OPT_{s^*}). 
		\end{align*}
		
Assume that $T'=\{i\}$ and $s(i)< 2s^*$. 
We observe that $f(\{i\}) \leq 2f(\OPT_{s^*})$ by Lemma \ref{lemma:det} and $s(i)/2 < s^*$. 
Note that all items in $\OPT_{C}\setminus\{i\}$ have size at most $s^*$ and their total size is at most $2s^*$. 
Thus we can divide $\OPT_{C}\setminus\{i\}$ into three sets $A_1, A_2, A_3$ with $s(A_j) \leq s^* \ (j=1,2,3)$. 
Hence, we have 
		\begin{align*}
		f(\OPT_C) \le f(\{i\})+f(\OPT_{C}\setminus\{i\}) \le 2f(\OPT_{s^*})+3f(\OPT_{s^*})
		= 5f(\OPT_{s^*}).
		\end{align*}

If $T'=\{i,i'\}$, then $s(i),s(i')<2s^*$ and $C-s(i)-s(i') < s^*$. 
In this case, 
by Lemma~\ref{lemma:det} and $s(i)/2,s(i')/2 < s^*$,
we have $f(\{i\}), f(\{i'\}) \leq 2f(\OPT_{s^*})$.
This implies that
		\begin{align*}
		f(\OPT_C)
		\le f(\{i\})+f(\{i'\})+f(\OPT_{C}\setminus\{i,i'\}) 
		\le (2+2+1)f(\OPT_{s^*})
=5f(\OPT_{s^*}). 
		\end{align*}

Finally, if $T'=\emptyset$, i.e., $s(i)\le s^*$ for all $i\in\OPT_{C}$, then 
we can divide $\OPT_C$ into five sets $A_1, \ldots, A_5$ with $s(A_j) \leq s^* \ (\forall j)$, and hence we have $f(\OPT_C)
		\leq f(A_1) + \cdots + f(A_5) 
		\le 5f(\OPT_{s^*})$. 
Therefore, the claim holds. 
\end{proof}

\begin{claim}
$f(\OPT_{s^*})\le \frac{3}{2\app}\cdot f(\{i^*\})$. 
\end{claim}
\begin{proof}
Recall that $f(\{i^*\})\ge 2 \app\cdot f(\OPT_{s^*/2})$.
Take an arbitrary item $i$ with the largest size in $\OPT_{s^*}$.

Assume that $s(i)>s^*/2$ and $i\in \SV$. 
We have $f(\{i\}) \leq \max \{ f(\{i'\}) \colon i' \in I_C \cap \SV \} = f(\{i^*\})$ since $i \in \SV$ and $s(i) \leq s^* \leq C$. 
Moreover, $f(\OPT_{s^*}\setminus\{i\})\leq f(\OPT_{s^*-s(i)})\leq f(\OPT_{s^*/2})$ follows from $s(i)> s^*/2$.
Hence, we have
$f(\OPT_{s^*})
\le f(\{i\})+f(\OPT_{s^*}\setminus\{i\})
\le f(\{i^*\})+f(\OPT_{s^*/2})
\le \left(1+\frac{1}{2\app}\right)f(\{i^*\})$. 

Suppose that $s(i)>s^*/2$ and $i\not\in \SV$. Then, 
 $f(\{i\})\le 2 f(\OPT_{s(i)/2})\le 2 f(\OPT_{s^*/2})$ holds by $i\not\in \SV$ and $s(i)\le s^*$,
and 
$f(\OPT_{s^*}\setminus\{i\})\leq  f(\OPT_{s^*/2})$ holds by $s(i)> s^*/2$.
These imply that $f(\OPT_{s^*})\le f(\{i\})+f(\OPT_{s^*}\setminus\{i\})
\le 2 f(\OPT_{s^*/2})+f(\OPT_{s^*/2})
\le \frac{3}{2\app}f(\{i^*\})$.   

If $s(i)\le s^*/2$, then we can divide $\OPT_{s^*}\setminus\{i\}$ into two sets $A_1,A_2$ such that $s(A_j)\le s^*/2$ $(j=1,2)$, and hence it follows that
$f(\OPT_{s^*})\le f(\OPT_{s(A_1)})+f(\OPT_{s(A_2)})+f(\OPT_{s(i)})\le 3f(\OPT_{s^*/2})\le \frac{3}{2\app}\cdot f(\{i^*\})$. 
Therefore, the claim follows since $3/(2\app) \geq 1+{1}/ ({2\app})$. 
\end{proof}

Combining these claims give
 $f(\mathcal{P}(C)) \geq (2\app/21) \cdot f(\OPT_C)$.
Hence the theorem is proven.
\end{proof}

\subsection{Randomized $(1-1/\sqrt[4]{e})/2$-robust universal policy}
\label{sec:rand_universal}

In this subsection, we devise a randomized $(1-1/{\sqrt[4]{e}})/2$-robust universal policy by modifying Algorithm \ref{alg:submodular_universal}. 
Note that we cannot use directly Algorithm \ref{alg:greedy1} in universal policies, because they do not use the observation while packing items so far. 
Instead we guess the capacity by the doubling strategy and emulate the execution of Algorithm \ref{alg:greedy1}. 
Our algorithm iteratively finds an item set $X$ maximizing $f(X)$ under a bound $C'$ on the total size $s(X)$, and then appends items of $X$ to the sequence. 
The bound $C'$ is set by the algorithm according to the input items. 
To compute the set $X$, we use Algorithm \ref{alg:greedy1} in which the capacity is set to be $C'$. 
We double the bound after each iteration. 

Also, we cannot use the other adaptive policy $\mathcal{P}^2$ in Algorithm \ref{alg:submodular_universal}. 
We remark that the proof of Theorem \ref{thm:unknown_rand_adaptive} uses only the fact that $f(\mathcal{P}^2(C))$ contains $i^C$ regarding to the policy $\mathcal{P}^2$. 
This implies that in the worst case analysis, there is no difference between $\mathcal{P}^2$ and a universal policy that inserts items based on the decreasing order according to the values of $f$. 
Thus, we replace $\mathcal{P}^2$ with the universal policy. 
Our algorithm is summarized in Algorithm \ref{alg:submodular_universal_rand_perm}. 

\begin{algorithm}[t]
\caption{Randomized $(1-1/\sqrt[4]{e})/2$-robust universal policy}\label{alg:submodular_universal_rand_perm}
\SetKwInOut{Input}{Input}
\SetKwInOut{Output}{Output}
$\Pi\ot ()$\;
flip a coin\;
\If{head}{
  $l\ot 1$, $s_{\min}\ot \min_{i\in I}s(i)$\;
  \For{$k\ot 0$ \KwTo $\lceil\log_2 (\sum_{i\in I}s(i)/s_{\min})\rceil$}{
    let $Y^{(k)}$ be the output $\mathcal{P}^1(2^{k} \cdot s_{\min})$ of the policy given in Algorithm \ref{alg:greedy1} for $(I, \emptyset)$\;
    \lForEach{$i\in Y^{(k)}\setminus\bigcup_{j=0}^{k-1}Y^{(j)}$}{$\Pi_l\ot i$, $l\ot l+1$}
  }
}
\lElse{
  let $\Pi$ be the decreasing order of items $i \in I$ in value $f(\{i\})$
 }
 \Return $\Pi$\;
\end{algorithm}

We remark that Algorithm \ref{alg:submodular_universal_rand_perm} constructs a sequence of items in polynomial time with respect to the input size. 
We can prove the following result by using Lemma \ref{lemma:submodular_greedy}. 

\begin{theorem}\label{thm:unknown_rand_univ}
	Algorithm \ref{alg:submodular_universal_rand_perm} is a $(1-1/{\sqrt[4]{e}})/2 > 0.110$-robust randomized universal policy.
\end{theorem}
\begin{proof}
	Let $\Pi^1$ (respectively, $\Pi^2$) be the sequence of items in $I$ returned by Algorithm \ref{alg:submodular_universal_rand_perm} when the coin comes up head (respectively, tail). 
	Suppose that the given capacity is $C$.
	The expected value of the output of Algorithm \ref{alg:submodular_universal_rand_perm} is $f(\ALG_C) = (f(\Pi^1(C)) + f(\Pi^2(C)))/2$. 
	We assume that $C\ge s_{\min}$ since otherwise $f(\OPT_C)=\ALG_C=0$. 
	Recall that $s_{\min} = \min_{i\in I}s(i)$. 
	Let $k$ be the number satisfying $2^{k-1} \cdot s_{\min}\le
 C<2^k \cdot s_{\min}$.
 We let $i_C \in \argmax \{f(\{i\}) \colon s(i)\leq C\}$.
	
	When $k=1$ (i.e., $C < 2 s_{\min}$), we have $\OPT_C=\{i^C\}$ because we can put only one item into the knapsack in this case.
	We also have $\Pi^2(C)=\{i^C\}$. 
	Hence, it holds that $f(\OPT_C)=f(\Pi^2(C))\le 2\cdot (f(\Pi^1(C)) + f(\Pi^2(C)))/2= 2\cdot f(\ALG_C)$.
	
	In what follows, we assume $k\ge 2$. 
	We observe that the total size of items in $\bigcup_{j=0}^{k-2}Y^{(j)}$ is at most $s\left(\bigcup_{j=0}^{k-2}Y^{(j)}\right)\le \sum_{j=0}^{k-2} s(Y^{(j)}) \leq \sum_{j=0}^{k-2} 2^j\cdot s_{\min}\le 2^{k-1}\cdot s_{\min} \leq C$. 
	Thus, all items in $\bigcup_{j=0}^{k-2}Y^{(j)}$, in particular
 those in $Y^{(k-2)}$, are contained in $\Pi^1(C)$. 
	We also observe that $f(\Pi^2(C)) \geq f(\{i^C\})$ because $i^C \in \Pi^2(C)$. 
	Hence, it holds that $2f(\ALG_C) = f(\Pi^1(C))+f(\Pi^2(C))  \ge f(Y^{(k-2)})+f(\{i^C\})$.

	We denote $C'=2^{k-2}\cdot s_{\min}$. 
	Recall that $Y^{(k-2)}$ is the output of the greedy algorithm for $(I, \emptyset)$ when the capacity is $C'$. 
	Let $(i_1,i_2,\dots,i_{n})$ be the greedy order for
 $(I,\emptyset)$ with capacity $C'$. 
	Let $q$ be the smallest index such that $i_q \in \OPT_C$ and $i_q \not\in Y^{(k-2)}$. 
	Since this definition implies $i_q \in I_C$,
	we have $f(\{i^C\}) \geq f(\{i_q\})$.
	By the monotonicity of $f$, 
	we also have $f(Y^{(k-2)})\geq f(Y^{(k-2)} \cap \{i_1,\dots,i_{q-1}\})$.
	Hence,
	\[
 f(\ALG_C)
	\ge \frac{f(Y^{(k-2)} \cap \{i_1,\dots,i_{q-1}\})+f(\{i_q\})}{2} \ge \frac{f((Y^{(k-2)} \cap \{i_1,\dots,i_{q-1}\}) \cup \{i_q\})}{2}.
	\]
	We evaluate $f((Y^{(k-2)} \cap \{i_1,\dots,i_{q-1}\}) \cup \{i_q\})$ using Lemma \ref{lemma:submodular_greedy}. 
	For notational convenience, we denote $Y' = (Y^{(k-2)} \cup \OPT_C) \cap \{i_1, \ldots, i_q\}$. 
	In a similar way to the proof of Theorem \ref{thm:unknown_rand_adaptive}, we can see that $(Y^{(k-2)} \cap \{i_1,\dots,i_{q-1}\}) \cup \{i_q\} = Y'$ by the choice of $q$. 
	Thus, we have $s(Y') > C'$, and it follows that $s(Y')/C \geq C'/C \geq (2^{k-2}\cdot s_{\min}) / (2^{k}\cdot s_{\min}) = 1/4$. 
	Therefore, by Lemma \ref{lemma:submodular_greedy}, we have $f(\ALG_C) \geq {f(Y')}/{2} \ge \left(1-{1}/{\sqrt[4]{e}}\right)f(\OPT_C)/2$.
\end{proof}

%% file: upperbound.tex
\section{Upper bounds on robustness ratios}
\label{sec:upperbound}

\subsection{Randomized policies for \knapsackUC with cancellation}
\label{sec:hardness_cancel}

In this subsection, we prove that no randomized policy achieves a robustness ratio better than $8/9$
even if the objective function is modular when cancellation is allowed.
It is known that the robustness ratio achieved by any deterministic policies for \knapsackUC is at most $1/2$~\cite{DisserKMS17},
but there was no upper bound on the robustness ratio for randomized policies.

Suppose that there are three items $a,b,c$ whose sizes are $2,3,4$, respectively, and
whose weights are equal to their own sizes.
Recall that the objective value is defined as the sum of the weights of
selected items.
Let $\OPT_{C'}$ be an optimal solution when the capacity is $C'\in\mathbb{R}_+$.
We provide an upper bound for this instance using Yao's principle \cite{yao1977}.
We consider an adversary that submits a probability distribution for the capacity as a mixed strategy. 
Let $C$ be the random variable which represents the capacity.
Then, the robustness ratio for this instance is upper-bounded by
\begin{align*}
  \max_{\mathcal{P}}\min_{C'\in\mathbb{R}_+}\frac{f(\mathcal{P}(C'))}{f(\OPT_{C'})}
  \le \max_{\mathcal{P}}\mathbb{E}_C\left[\frac{f(\mathcal{P}(C))}{f(\OPT_{C})}\right],
\end{align*}
where the maximum is taken over all randomized policies $\mathcal{P}$.

We assume that the adversary submits $C=4$ with probability $4/9$ and $C=5$ with probability $5/9$.
Then, we have
 \begin{align*}
  \max_{\mathcal{P}}\mathbb{E}_C\left[\frac{f(\mathcal{P}(C))}{f(\OPT_{C})}\right]
  =\max_{\mathcal{P}}\left(\frac{4}{9}\cdot\frac{f(\mathcal{P}(4))}{f(\OPT_{4})}+\frac{5}{9}\cdot\frac{f(\mathcal{P}(5))}{f(\OPT_{5})}\right).
 \end{align*}
   Note that  $f(\OPT_4)=4$ and $f(\OPT_4)=5$.
   Hence the right-hand side is equal to
   $\max_{\mathcal{P}}(f(\mathcal{P}(4))+f(\mathcal{P}(5)))/9$.
   In what follows, we prove that this is at most $8/9$.
   Note that the maximum is attained by a deterministic policy.
If a deterministic policy $\mathcal{P}$ chooses item $a$ first, then it can get objective value $2$ when the capacity is $4$.
Thus, we have $f(\mathcal{P}(4))+f(\mathcal{P}(5))\le 2+5=7$.
Similarly, the value is at most $3+5=8$ if it selects $b$ first
and the value is at most $4+4=8$ if it selects $c$ first.
Since $f(\mathcal{P}(4))+f(\mathcal{P}(5))\le \max\{7, 8, 8\} = 8$, we obtain the following theorem. 

\begin{theorem}
Even if cancellation is allowed,
no randomized policy has a robustness ratio better than $8/9 > 0.888$
for \knapsackUC.
\end{theorem}

\subsection{Randomized policies for \knapsackUC without cancellation}
\label{sec:hardness_nocancel}

In this subsection, we prove that no randomized policy achieves a constant
robustness ratio
for \knapsackUC
in the setting where cancellation is not allowed.

Let $M$ be a positive integer of at least 2.
We assume that the item set $I$ consists of $n$ items $1,\ldots,n$, and the size and the weight of item $i$ is $M^i$.
The objective function $f$ is defined by $f(S)=\sum_{i \in S}M^i$ for any $S \subseteq I$.
We fix a policy, and show that the robustness ratio of this policy is $O(1/M)$ for this instance if $n = \Omega(M)$.

Since there are $n$ items, at least one item is chosen first by the policy with probability at most $1/n$.
Let $i$ denote such an item. 
Let us consider the case where the capacity is $M^i$.
When the policy does not choose $i$ first (this happens with probability at least $(n-1)/n$), 
the largest objective value achieved by the solution is $\sum_{i'=1}^{i-1}M^{i'} \leq 2M^{i-1}$.
Therefore, the expected objective value of the policy is at most $M^i / n + (n-1)/n \cdot 2M^{i-1}$.
On the other hand,
the optimal solution for this instance consists of only item $i$, which attains the objective value $M^i$.
The gap between these values is $O(1/M)$ if $n=\Omega(M)$.

\begin{theorem}
If cancellation is not allowed,
then no randomized policy achieves a constant robustness ratio for \knapsackUC.
\end{theorem}

\subsection{Deterministic policies for the unit size case of \submaxUC}
\label{app.unitcase}
In this subsection, we consider a special case of \submaxUC in which the
size of each item is $1$, i.e., the cardinality constraint case.
We show that no deterministic policy (even one with no computational
restrictions) achieves a robustness ratio better than $(1+\sqrt{5})/4~(>
0.809)$,
and no randomized policy achieves a robustness ratio better than
$(5+\sqrt{5})/8~(>0.904) $ for this case.
We remark that the cancellation is not useful
in the cardinality constraint case.
Also, for the cardinality constraint case of \submaxUC,
a greedy algorithm achieves $(1-1/e)$-robust (-approximation) and it is
known to be the best possible among policies that run in polynomial time.

To present the upper bound on the robustness ratio,
let us construct an instance of the problem.
Suppose that there are three items $a,b,c$, each of whose size is $1$, and the objective function $f$ is given by
\begin{align*}
&f(\emptyset)=0,~f(\{a\})=4,~f(\{b\})=f(\{c\})=1+\sqrt{5},\\
&f(\{a,b\})=f(\{a,c\})=3+\sqrt{5},~f(\{b,c\})=f(\{a,b,c\})=2+2\sqrt{5}.
\end{align*}
Note that the function is monotone submodular and symmetric in $b$ and $c$.
If the policy first packs $a$, then the robustness ratio for capacity $1$ is equal to $f(\{a\}/f(\{a\}) = 1$, and 
the one for capacity $2$ is ${f(\{a,b\})}/{f(\{b,c\})}={(3+\sqrt{5})}/{(2+2\sqrt{5})}={(1+\sqrt{5})}/{4}$. 
Otherwise, i.e., the policy first packs $b$ or $c$,
then the robustness ratio for capacity $1$ is 
$\frac{f(\{b\})}{f(\{a\})}={(1+\sqrt{5})}/{4}$, and the one for capacity $2$ is at least $f(\{a, b\})/f(\{ b, c\}) = {(1+\sqrt{5})}/{4}$. 
Thus, no deterministic policy achieves a robustness ratio better than $(1+\sqrt{5})/4$.

Finally, we prove that no randomized policy achieves a robustness ratio better than $(5+\sqrt{5})/8$.
Let us consider a randomized policy for the above instance that first inserts $a$ with probability $p$.
Then the robustness ratio for capacity $1$ is 
\begin{align*}
\frac{p\cdot f(\{a\})+(1-p)\cdot f(\{b\})}{f(\{a\})}
=\frac{4p+(1-p)(1+\sqrt{5})}{4}
=\frac{(3-\sqrt{5})p+(1+\sqrt{5})}{4}.
\end{align*}
Also, the robustness ratio for capacity $2$ is at most
\begin{align*}
\frac{p\cdot f(\{a,b\})+(1-p)\cdot f(\{b,c\})}{f(\{b,c\})}=\frac{(3+\sqrt{5})p+(2+2\sqrt{5})(1-p)}{(2+2\sqrt{5})}=\frac{4-(3-\sqrt{5})p}{4}.
\end{align*}
Note that the former value is monotone increasing for $p$, and the latter is monotone decreasing for $p$. 
Thus, the robustness ratio of the policy is at most
\begin{align*}
\min\left\{\frac{(3-\sqrt{5})p+(1+\sqrt{5})}{4},\frac{4-(3-\sqrt{5})p}{4}\right\}\le\frac{5+\sqrt{5}}{8}, 
\end{align*}
where the inequality holds when $p=1/2$.

\begin{theorem}
 For \submaxUC with only unit-size items,
 no deterministic policy achieves a robustness ratio better than
 $(1+\sqrt{5})/4$,
 and no randomized policy achieves a robustness ratio better than $(5+\sqrt{5})/8$.
\end{theorem}

%% file: stochastic.tex
\section{Approximation algorithms for \submax without cancellation}
\label{sec.stochastic}

\subsection{Pseudo-polynomial time $(1/4-o(1))$-approximation algorithm}
\label{subsec.pseudo-stochastic}

We present a pseudo-polynomial time $(1/4-o(1))$-approximation algorithm for \submax without cancellation.
We reduce the problem to the following problem.

\smallskip
\emph{Submodular maximization problem with an interval independent
constraint: }
We are given a set $I$ of items.
Each item $i \in I$ is associated with an interval $l_i$
on a line. 
We are also given a submodular function $f\colon 2^I
\rightarrow \Rset_+$.
The objective is to find a subset $I'$ of $I$
maximizing $f(I')$
subject to the constraint that no two intervals associated with items in
$I'$ intersect, i.e., $l_i \cap l_j = \emptyset$ for all $i,j \in I'$.

\smallskip
Feldman~\cite{Feldmanphd} showed that this problem admits a ($1/4-o(1)$)-approximation randomized algorithm for monotone
submodular functions.

Let us explain the reduction from \submax to
this problem.
Let $I$ be the set of items and $f \colon 2^I \rightarrow \Rset_+$
be the submodular function given in an instance of \submax.
Recall that $T=s(I)$ and $[T'] =\{0,1,\ldots,T'\}$.
For each $i \in I$,
we make $T-s(i)+1$ copy items $i_0,\ldots,i_{T-s(i)}$,
and $i_j$ is associated with the interval $[j,j+s_i-1]$ for each $j \in [T-s_i]$.
Let $I'$ denote the set of these items.
We define a function $f' \colon 2^{I'} \rightarrow \Rset_+$ 
by
\[
 f'(U') = \sum_{t=0}^T p(t) f(\{i \in I
 \colon \exists j \in [t-s_i], \ i_j \in U' \})
\]
for all $ U' \subseteq I'$.
It is not difficult to prove the following lemma.

\begin{lemma}
The function $f'$ is monotone and submodular.
\end{lemma}

Let $U' \subseteq I'$ be a solution 
for the instance of the problem with an interval independent constraint
that consists of the item set $I'$ (with associated intervals) and the
submodular function $f'$.
From $U'$, we define the ordering of $I$ as follows.
If a copy of $i \in I$ is included in $U'$,
then pick item $i$ at the time equal to the minimum index of its copies
in $U'$, i.e, $\min\{t \in [T] \colon i_t \in U'\}$.
Sort the items in the increasing order of 
the times at which they are picked.
The other items follow these items, and their order
is decided arbitrarily.
This sequence achieves the objective value of at least $f'(U')$.

\begin{theorem}
 Problem \submax without cancellation 
 admits a
 pseudo-polynomial time randomized $(1/4-o(1))$-approximation algorithm for monotone submodular functions.
\end{theorem}

Feldman~\cite{Feldmanphd} also gave a 
 randomized $1/(2e+o(1))$-approximation algorithm for 
 the problem with an interval independent constraint
 and nonmonotone submodular functions.
 Thus, if the submodular function is not monotone, then \submax admits 
 a pseudo-polynomial time randomized $1/(2e+o(1))$-approximation algorithm.

\subsection{Polynomial-time $((1-1/\sqrt[4]{e})/4-\epsilon)$-approximation algorithm}
\label{sec:poly_16approx}

In this subsection,
we present a polynomial-time algorithm of approximation ratio
$(1-1/\sqrt[4]{e})/4-\epsilon$ for any small constant $\epsilon >0$.
This algorithm is based on the idea of Gupta et al.~\cite{GuptaKMR11}
for the stochastic knapsack problem.
We first give a pseudo-polynomial time $((1-1/\sqrt[4]{e})/2-o(1))$-approximation algorithm,
and then we transform it into a polynomial-time algorithm.

Our algorithm relies on a continuous relaxation of the problem. 
The relaxation is formulated based on an idea of using time-indexed variables;
we regard the knapsack capacity as a time limit while considering that picking an item $i$ spends time $s(i)$.
In the relaxation, 
we have a variable $x_{ti} \in [0,1]$ for each $t \in [T-1]$ and $i \in I$,
and $x_{ti}=1$ represents that item $i$ is picked at time $t$.
 For each $t \in [T]$ and $i\in I$,
let $\bar{x}_{ti}=\sum_{t'\in [t-s(i)]}x_{t'i}$ if $t \geq s(i)$, and let $\bar{x}_{ti} = 0$ otherwise.
For each $t\in [T]$, let $\bar{x}_t$ be the $|I|$-dimensional vector whose component corresponding to $i \in I$ 
is  $\bar{x}_{ti}$.
Let $F\colon [0,1]^I \rightarrow \Rset_+$ be the multilinear extension of the submodular function $f$.
Then, the relaxation is described as 
\begin{equation}
\label{eq.lpnaive-loose}
 \begin{array}{lll}
  \text{maximize} & \bar{F}(x):=\sum_{t=1}^T p(t)F(\bar{x}_t) &\\
  \text{subject to} 
  & \sum_{t \in [T-1]} x_{ti} \leq 1, & \forall i \in I,\\
  & \sum_{i \in I} \sum_{t'\in [t]} x_{t'i} \min\{s(i), t\}\leq 2t, & \forall t=1,\ldots,T,\\
  & x_{ti} \geq 0, & \forall t \in [T-1],  \forall i \in I.
 \end{array}
\end{equation}

Let us see that 
\eqref{eq.lpnaive-loose} relaxes the problem.
It is not difficult to see that the first and the third constants are valid.
We prove that the second constraint is valid. 
Suppose that $x$ is an integer solution that corresponds to a sequence of items.
Let $I'$ be the set of items picked at time $t$ or earlier in this solution.
Notice that $\sum_{i \in I} \sum_{t'\in [t]} x_{t'i} \min\{s(i), t\}
= \sum_{i \in I'} \min\{s(i),t\}$ holds.
Let $j$ be the item picked latest in $I'$.
Then, since the process of all items in $I'\setminus \{j\}$ terminates by time $t$,
we have $\sum_{i \in I' \setminus \{j\}}s(i) \leq t$.
Therefore,
$\sum_{i \in I'} \min\{s(i),t\} = \min\{s(j),t\} + \sum_{i \in I'\setminus \{j\}}s(i)
\leq 2t$.

Note also that $\bar{F}$ 
is a smooth monotone submodular function;
i.e., 
$\partial \bar{F}(x)/\partial x_{ti}\geq 0$
for any $t \in [T-1]$ and $i \in I$,
and 
$\partial^2 \bar{F}(x)/(\partial x_{ti} \partial x_{t'i'})\leq 0$
for any $t,t'\in [T-1]$ and $i,i'\in I$.
Hence, we can apply the continuous greedy algorithm for solving
\eqref{eq.lpnaive-loose}.
Let $x^*$ be an obtained feasible solution for \eqref{eq.lpnaive-loose}.
We first present a rounding algorithm for this solution.
Since the formulation size of this relaxation is not polynomial, 
this part does not run in 
polynomial time. We convert the algorithm into a polynomial-time one later.

\para{Rounding algorithm}
The algorithm consists of two rounds.
In the first round, each item $i$ chooses an integer $t$ from $[T-1]$
with probability $x^*_{ti}/4$,
and chooses no integer with probability $1-\sum_{t \in [T-1]}x^*_{ti}/4$.
An item is discarded if it chooses no integer.
Let $I_1$ be the set of remaining items.
For each $i \in I_1$, let $t_i$ denote the integer chosen by $i$.

Then, the algorithm proceeds to the second round. 
For each $i \in I_1$, let $J_i$ denote $\{j \in I_1 \colon t_j \leq t_i\}$.
In the second round, item $i$ is discarded if $s(J_i) \geq t_{i}$.
Let $I_2$ denote the set of items remaining after the second round.
The algorithm outputs a sequence 
obtained by sorting the items $i \in I_2$ in the non-decreasing order of $t_i$, 
where ties are broken arbitrarily, 
and by appending the other items after those in $I_2$ in an arbitrary order.

For $t \in [T]$,
let $I_t=\{i \in I_2 \colon t_i\leq t-s(i)-1\}$.
If $i \in I_t$, then $i$ contributes to the objective value of the solution
when the knapsack capacity is at least $t$.

\begin{lemma}
 \label{lem.rounding}
 For any $t \in [T]$, $I_t$ is the output
 of a monotone $(1/4,1/2)$-balanced contention resolution scheme
 for the maximization problem of $f$ under the knapsack capacity $t$
 and the fractional solution $\bar{x}^*_t$.
 Hence, the sequence output by the algorithm achieves an objective value of
 at least $\bar{F}(x^*/4)/2 $ in expectation.
\end{lemma}
 \begin{proof}
  Take arbitrarily $t \in [T]$.
  We define a random mapping $\pi\colon 2^I\rightarrow 2^I$
  as follows.
  Let $I' \subseteq I$.
  We let each $i \in I'$ independently sample an integer $t'_i$
  from $[t-s(i)-1]$
  with probability $x^*_{t'_i i}/\bar{x}^*_{ti}$.
  Define $J'_i=\{j \in I' \colon t'_{j} \leq t'_i\}$ for each $i \in I'$.
  Then, $\pi(I')$ is defined as $\{i \in I' \colon s(J'_i) < t'_i\}$.

  Let us see that $\pi$ is a $(1/4,1/2)$-balanced contention resolution scheme
  for $\bar{x}^*_t$.
  For this, we analyze the probability that 
  $i$ is included in $\pi(R_{\bar{x}^*_t/4})$,
  conditioned that $i \in R_{\bar{x}^*_t/4}$.
  Recall that $i\in R_{\bar{x}^*_t/4}$ 
is not included in $\pi(R_{\bar{x}^*_t/4})$ if
 $s(J'_i) \geq t'_i$ 
 Let $j$ be an arbitrary item other than $i$, and let $s'(j) = \min\{s(j),t'_i\}$.
 Notice that
 $s'(J'_i) \geq t'_i$ holds
 if $s(J'_i) \geq t'_i$ holds.
 We give  an upper bound on the probability that 
 $s'(J'_i) \geq t'_i$ happens.
 The item $j \in I \setminus \{i\}$ is included in 
  $R_{\bar{x}^*_t/4}$
  with probability $\bar{x}^*_{tj}/4$,
  and then
  it is included in $J'_i$ 
  (i.e., $j$ chooses an integer at most $t'_i$)
 with probability at most $\sum_{t' \in [t'_i]}x^*_{t'j}/\bar{x}^*_{tj}$.
 Hence $\mathbb{E}[s'(J'_i)] \leq \sum_{j\in I}  \sum_{t' \in [t'_i]}
 x^*_{t'j}\min\{s(j),t'_i\}/4 \leq t'_i/2$, where the last inequality
 follows from the second constraint of \eqref{eq.lpnaive-loose}.
 Applying Markov's inequality, we obtain $\Pr[s'(J'_i) \geq
 t'_i] \leq 1/2$.
  Therefore, $\Pr[i \in \pi(R_{\bar{x}^*_t/4}) \mid i \in R_{\bar{x}^*_t/4}] \geq 1/2$,
which means that $\pi$ is a $(1/4,1/2)$-balanced contention resolution scheme
for $\bar{x}^*_t$.
  
  Next, we prove that $\pi$ is monotone.
  Let $I' \subseteq I'' \subseteq I$.
  Then, $\sum_{j \in J'_i}s(i')$ is not smaller in the computation of $\pi(I'')$
  than in the computation of $\pi(I')$, if 
  each item in $I'$ samples the same integer both in $\pi(I')$  and $\pi(I'')$.
  This implies $\Pr[i \in \pi(I')] \geq \Pr[i \in \pi(I'')]$ for each $i \in I'$.
  Thus, $\pi$ is monotone.

  Let us observe that $I_t$ coincides with $\pi(R_{\bar{x}^*_t/4})$.
  Recall that,
  in the first round of the algorithm,
  each item $i$ independently chooses an integer $t_i$.
  The probability that $t_i \leq t-s(i)-1$ 
  is $\sum_{t' \in [t-s(i)-1]}x^*_{ti}/4=\bar{x}^*_{ti}/4$.
  Hence, 
  item set $\{i \in I \colon t_i \leq t-s(i)-1\}$ coincides with $R_{\bar{x}^*_t/4}$.
  Then, 
  the decision of whether or not an item $i$ in this set is discarded
  in the second round of the algorithm is the same 
  as the computation of $\pi(R_{\bar{x}^*_t/4})$.
  Therefore, $I_t$
  coincides with $\pi(R_{\bar{x}^*_t/4})$.
 \end{proof}

By Theorem~\ref{thm.contscheme} and Lemma~\ref{lem.rounding}, 
our algorithm achieves
$((1-1/\sqrt[4]{e})/2-o(1))$-approximation
if it is combined with the continuous greedy algorithm with
stopping time $1/4$. 

Lemma~\ref{lem.rounding} also implies that the integrality gap of \eqref{eq.lpnaive-loose} 
is at least 1/8. On the other hand, 
there exist some instances indicating that the integrality gap is at most $1/3+\epsilon$ for any $\epsilon > 0$
even when the objective function is modular.
Suppose that the capacity is $T > 1$ with probability 1, and there are three items:
two items $i$ and $j$ of size $T$, and an item $k$ of size $1$. The
weight of these items are all 1, and the objective value is defined as
the sum of the weights of chosen items.
Clearly a knapsack of capacity $T$ can include at most one of the items,
and hence the maximum objective value of integer solutions is 1.
On the other hand, 
a fractional solution defined by $x_{0i}=1$, $x_{0j}=(T-1)/T$, $x_{T-1,k}=1$ and setting the other variables to be $0$ achieves the objective value $3-1/T$.

\para{Conversion into a polynomial-time algorithm}
We transform the pseudo-polynomial algorithm 
into the polynomial-time one.
Let $W = f(I)$ and $w=\min_{i \in I}f(\{i\})$.
We assume $w > 0$ without loss of generality;
if $f(\{i\}) = 0$ for some item $i \in I$, we can safely remove $i$ from
$I$
because the submodularity implies $f(S)=0$ for any $S\subseteq I$ with $i\in S$.
Recall that we assume that the submodular function $f$ is given
as an oracle. Indeed, algorithms in this paper
can be implemented if we can compute the value of the function (or
the value of its multilinear extension). 
We assume that the oracle
is encoded in $\Omega(\log (W/w))$,
and hence we say that an algorithm runs in polynomial time
if its running time is expressed as a polynomial in $\log (W/w)$.

The idea for the conversion is to use a more compact relaxation, which
is obtained by defining
variables and constraints for 
a polynomial number of integers in $[T]$.
Let $\epsilon$ be a positive constant smaller than 1,
and let $\eta = \lfloor \log_{1-\epsilon}(\epsilon w /(W\log T))\rfloor$.
For each $t \in [T-1]$, let $\bar{p}(t)$ denote $\sum_{t'=t+1}^T p(t')$.
We first define $\{\tau_0,\tau_1,\ldots,\tau_{q},\tau_{q+1}\}\subseteq [T]$ 
such that
$q=O(\log T+\log(W/(w\epsilon)))$,
$\tau_0=0$, $\tau_1=1$, $\tau_{q+1}=T$,
$\tau_{j} < \tau_{j+1} \leq 2\tau_j$ holds for any $j=1,\ldots,q$,
and there exists $q_{\eta} \in \{1,\ldots,q\}$ satisfying the following conditions:
   \begin{itemize}
\item $\bar{p}(\tau_{j}) \geq \bar{p}(\tau_{j+1}-1) \geq (1-\epsilon) \bar{p}(\tau_j)$ for any $j \in [q_{\eta}]$;
\item $\bar{p}(\tau_j) < \epsilon w/(W \log T)$ for any $j \in \{q_{\eta}+1,\ldots,q\}$.
 \end{itemize}
Such a subset of $[T]$ can be defined as follows.
For $j \in \{1,\ldots, \eta+1\}$, let $\tau'_j$ be the
minimum integer $t\in [T-1]$ such that $\bar{p}(t)< (1-\epsilon)^{j-1}$.
We assume without loss of generality that $\bar{p}(\min_{i \in I}s(i))=1$,
which means $\tau'_1 \geq \min_{i \in I}s(i)$.
We denote the set of positive integers in
$ \{\tau'_j  \colon j=1,\ldots,\eta+1\} 
\cup \{2^j \colon j \in [\lceil\log T \rceil -1]\}$
by $\{\tau_1,\ldots,\tau_{q}\}$,
and sort those integers so that $1=\tau_1 < \tau_2 < \cdots <
\tau_{q}$. 
We define $q_{\eta}$ so that $\tau_{q_{\eta}}=\tau_{\eta+1}'$.
Then, the obtained subset satisfies the above conditions.

In addition, we define the set of integers in $\{\tau_0,\ldots,\tau_{q+1}\}\cup \{\tau_j -s(i)+1 \colon j \in \{1,\ldots,q+1\}, i \in I\}$ 
as $\{\xi_0,\ldots,\xi_{r+1}\}$, where $0=\xi_0<\xi_1 < \ldots < \xi_r< \xi_{r+1}=T$. 
Notice that $r=O(n\log T + n\log (W/(w\epsilon)))$.

Roughly speaking, we define variables for each $\xi_k$ ($k \in [r]$),
and constraints
for each $\tau_j$ ($j \in [q]$).
Specifically, a variable $y_{k i}$ is defined for each $k \in [r]$ and $i \in I$,
and $y_{k i}$ replaces variables $x_{\xi_k, i},\ldots,x_{\xi_{k+1} - 1, i}$ 
in \eqref{eq.lpnaive-loose}.
For $j=1, \ldots,q+1$ and $i \in I$, we define 
an auxiliary variable $z_{j i}$ as $\sum_{\xi_{k+1}-1\leq \tau_j-s(i)} y_{ki}$,
and define $z_{j}$ as the $|I|$-dimensional vector whose component corresponding to $i \in I$ is $z_{j i}$.
Then, the compact relaxation is described as follows.
\begin{equation}
\label{eq.lpcompact}
 \begin{array}{lll}
  \text{maximize} & \sum_{j=0}^{q}\bar{p}(\tau_{j}) (F(z_{j+1}) - F(z_{j}))\\
  \text{subject to} &
   \sum_{k \in [r]} y_{k i} \leq 1, & \forall i \in I,\\
  & \sum_{i \in I}\sum_{\xi_k < \tau_j} y_{ki} \min\{s(i), \tau_j\}\leq 2\tau_j, & \forall j \in \{1,\ldots,q\},\\
  & z_{ji} = \sum_{\xi_{k+1} -1 \leq \tau_j - s(i)} y_{ki} & \forall j \in [q],\\
  & y_{k i} \geq 0, & \forall i \in I, \forall k \in [r].
 \end{array}
\end{equation}

\begin{lemma}
\label{lem.relaxation}
The optimal objective value of \eqref{eq.lpcompact}
is not smaller than that of \eqref{eq.lpnaive-loose}.
\end{lemma}
\begin{proof}
Suppose that $x$ is a feasible solution for  \eqref{eq.lpnaive-loose}.
From $x$, we define 
a solution $y$ for \eqref{eq.lpcompact}
so that $y_{k i}=\sum_{t =\xi_k}^{\xi_{k+1}-1} x_{ti}$ for each $k \in [r]$ and $i \in I$.
We define variables $z$ from $y$ by the third constraints of \eqref{eq.lpcompact}. 
Then, $(y,z)$ is feasible to \eqref{eq.lpcompact}.
Indeed, it is immediate from the feasibility of $x$ in \eqref{eq.lpnaive-loose} that $(y,z)$
satisfies the constraints of \eqref{eq.lpcompact} except the second one.
As for the second constraints, we can observe
that $\sum_{\xi_k < \tau_j}y_{ki}=\sum_{t < \tau_j}x_{ti}$ holds for any $i \in I$ and $j \in \{1,\ldots,q\}$
by the definition of $y$
and the fact that $\tau_j$ is included in $\{\xi_1,\ldots,\xi_{r}\}$.

Let us show that the objective value of $(y,z)$ in \eqref{eq.lpcompact}
 is not smaller than $\bar{F}(x)$. We have
\[
 \bar{F}(x)  
  =\sum_{t=1}^T p(t)F(\bar{x}_t)
  = \bar{p}(T-1)F(\bar{x}_T)+\sum_{t=1}^{T-1} (\bar{p}(t-1)-\bar{p}(t)) F(\bar{x}_t)
  = \sum_{t=0}^{T-1} \bar{p}(t) (F(\bar{x}_{t+1})- F(\bar{x}_{t})),
\]
where $\bar{x}_0$ denotes the zero-vector for convention.
The right-hand side can be written as 
\begin{align}
 \sum_{t=0}^{T-1} \bar{p}(t) (F(\bar{x}_{t+1})- F(\bar{x}_{t}))
& = 
\sum_{j=0}^q \sum_{t=\tau_j}^{\tau_{j+1}-1} \bar{p}(t) (F(\bar{x}_{t+1})- F(\bar{x}_{t}))
\nonumber
\\
& \leq 
\sum_{j=0}^q \sum_{t=\tau_j}^{\tau_{j+1}-1} \bar{p}(\tau_j) (F(\bar{x}_{t+1})- F(\bar{x}_{t}))
\nonumber
\\
& =
\sum_{j=0}^q \bar{p}(\tau_j) (F(\bar{x}_{\tau_{j+1}})- F(\bar{x}_{\tau_j})).
\label{eq.bound-lp}
\end{align}
Recall that $\bar{x}_{\tau_j i}=\sum_{t\leq \tau_j-s(i)}x_{ti}$
for each $j \in [q]$ and $i\in I$.
There exists $k' \in [r]$ such that $\tau_j-s(i)+1=\xi_{k'}$,
and hence 
$\sum_{t\leq \tau_j-s(i)}x_{ti}$ can be written as $\sum_{\xi_{k}-1\leq \tau_j-s(i)}y_{ki}=z_{ji}$.
Thus $\bar{x}_{\tau_j}=z_{j}$ holds for each $j \in [q]$,
implying that \eqref{eq.bound-lp} is equal to the objective value of $(y,z)$.
\end{proof}

\begin{lemma}
\label{lem.conversion}
From a feasible solution to \eqref{eq.lpcompact} achieving the objective value $\theta$,
we can construct a feasible solution to \eqref{eq.lpnaive-loose} achieving
the objective value of at least $(1-\epsilon)(\theta-\epsilon w)/2$.
\end{lemma}
\begin{proof}
Let $(y,z)$ be a feasible solution to \eqref{eq.lpcompact}.
For each $k \in [r]$, $i \in I$, and $t \in \{\xi_{k},\ldots,\xi_{k+1}-1\}$,
we define $x_{ti}$ as $y_{ki}/(\xi_{k+1}-\xi_k)$.

 We prove that $x/2$ is feasible to \eqref{eq.lpnaive-loose}.
 It is immediate from the definition of $x$ that $x$ satisfies
 the first and the third constraints of \eqref{eq.lpnaive-loose}.
 We focus on the second constraints.
 Let $t \in \{1,\ldots,T-1\}$.
 Suppose that $\tau_{j} \leq t < \tau_{j+1}$ holds for some $j \in [q]$.
 Then, 
 for each $i \in I$, we have 
$\sum_{t' \in [t]} x_{t'i} \leq \sum_{t' < \tau_{j+1}} x_{t'i}=\sum_{\xi_k < \tau_{j+1}}y_{ki}$,
 where the equality follows from the fact that $\tau_{j+1} \in \{\xi_1,\ldots,\xi_{r+1}\}$.
Moreover, $\min\{s(i),t\} \leq \min\{s(i),\tau_{j+1}\}$ also holds.
 Hence,
 $$
 \sum_{i \in I} \sum_{t' \in [t]}x_{t'i}\min\{s(i),t\} \leq
 \sum_{i\in I}\sum_{\xi_k < \tau_{j+1}} y_{ki}\min\{s(i), \tau_{j+1}\} \leq
 2\tau_{j+1}
 $$
 holds, where the last inequality follows from the second constraints of
 \eqref{eq.lpcompact}
 when $j < q$, and from $\tau_{q+1}=T=\sum_{i \in I}s(i)$ and the first
 constraints of  \eqref{eq.lpcompact} when $j=q$.
 By its definition, 
 $t \geq \tau_j \geq \tau_{j+1}/2$.
 This implies that $x/2$ is feasible to \eqref{eq.lpnaive-loose}.

 Let $\theta$ be the objective value of $(y,z)$ in \eqref{eq.lpcompact}.
 We show that the objective value $\bar{F}(x/2)$ of $x/2$ in \eqref{eq.lpnaive-loose} is at least $(1-\epsilon)(\theta-\epsilon w)/2$.
We observe that
 \[
  \bar{F}\left(\frac{x}{2}\right)
 = \sum_{t=1}^T p(t)F\left(\frac{\bar{x}_t}{2}\right)
\geq \frac{1}{2} \sum_{t=1}^T p(t)F\left(\bar{x}_t\right)
=
\frac{1}{2} \sum_{j=0}^q \sum_{t=\tau_j}^{\tau_{j+1}-1} \bar{p}(t) \left(F\left(\bar{x}_{t+1}\right)- F\left(\bar{x}_{t}\right)\right).
 \]
Since $\bar{p}(t) \geq \bar{p}(\tau_{j+1}-1)$ for any $t \leq \tau_{j+1}-1$,
the right-hand side of the above inequality satisfies
\begin{align*}
 \frac{1}{2} \sum_{j=0}^q \sum_{t=\tau_j}^{\tau_{j+1}-1} \bar{p}(t) \left(F\left(\bar{x}_{t+1}\right)- F\left(\bar{x}_{t}\right)\right)
& \geq
 \frac{1}{2} \sum_{j=0}^q  \bar{p}(\tau_{j+1}-1) \sum_{t=\tau_j}^{\tau_{j+1}-1}\left(F\left(\bar{x}_{t+1}\right)- F\left(\bar{x}_{t}\right)\right)\\
& =
 \frac{1}{2} \sum_{j=0}^q  \bar{p}(\tau_{j+1}-1) \left(F\bigl(\bar{x}_{\tau_{j+1}}\bigr)- F\bigl(\bar{x}_{\tau_j}\bigr)\right).
 \end{align*}
Recall that  $\bar{x}_{\tau_j}=z_j$ and $\bar{p}(\tau_{j+1}-1) \geq (1-\epsilon)\bar{p}(\tau_j)$ hold for any $j \in [q_{\eta}]$.
Therefore,
\[
  \frac{1}{2} \sum_{j=0}^q  \bar{p}(\tau_{j+1}-1) \left(F\bigl(\bar{x}_{\tau_{j+1}}\bigr)- F\bigl(\bar{x}_{\tau_j}\bigr)\right)
\geq 
  \frac{1-\epsilon}{2} \sum_{j=0}^{q_{\eta}} \bar{p}(\tau_{j}) \left(F\left(z_{j+1}\right)- F\left(z_j\right)\right).
\]

On the other hand, $\theta$ can be written as 
\[
 \theta  
 =  \sum_{j=0}^{q} \bar{p}(\tau_j) (F(z_{j+1}) - F(z_j))
 =  \sum_{j=0}^{q_{\eta}} \bar{p}(\tau_j) (F(z_{j+1}) - F(z_j))
 +  \sum_{j=q_{\eta}+1}^{q} \bar{p}(\tau_j) (F(z_{j+1}) - F(z_j)).
\]
Recall that $\bar{p}(\tau_j) \leq \epsilon w/(W\log T)$ if $j \geq q_{\eta}+1$.
Moreover, we have $q-q_{\eta}\leq \log T$, and hence
\[
  \sum_{j=q_{\eta}+1}^{q} \bar{p}(\tau_j) (F(z_{j+1}) - F(z_j))
\leq 
  \sum_{j=q_{\eta}+1}^{q} \bar{p}(\tau_j) W
\leq \epsilon w.
\]
Combining all these discussion, we have
$ \bar{F}(x/2)
 \geq (1-\epsilon) (\theta-\epsilon w)/2$.
\end{proof}

We now wrap up our algorithm.
Our algorithm
first applies the continuous greedy algorithm with stopping time $1/4$
to compute a solution
$y$
such that $4y$ is feasible for 
\eqref{eq.lpcompact}
and the objective value of $y$ in \eqref{eq.lpcompact}
is $1-1/\sqrt[4]{e}$ times that of any feasible solution, particularly the optimal value $h$ of \eqref{eq.lpcompact}. 
From $y$, we compute a solution $x$ for \eqref{eq.lpnaive-loose} by
Lemma~\ref{lem.conversion}.
We see that $4x$ is feasible for \eqref{eq.lpnaive-loose}, 
and the objective value of $x$ in \eqref{eq.lpnaive-loose}
is at least $(1-\epsilon)((1-1/\sqrt[4]{e})h-\epsilon w)/2$.
Since we are assuming $\bar{p}(\min_{i\in I}s(i)) \geq 1$,
picking the item of the smallest size at time 0
achieves objective value $w$.
This means $h \geq w$,
and hence the objective value of $x$
is at least
$(1-\epsilon)(1-\epsilon - 1/\sqrt[4]{e})h/2$.
Then, applying the rounding algorithm to $4x$,
we obtain a sequence of objective value 
$(1-\epsilon)(1-1/\sqrt[4]{e}-\epsilon)h/4$.

To make this algorithm run in polynomial-time,
we do not explicitly write down $x$.
In the rounding algorithm, values of $x$ are used
for deciding $t_i$ for each $i \in I$ in the first round of the rounding algorithm.
This is possible without writing down $x$ as follows.
Notice that $x_{ti}$ takes the same value for any $t \in [\tau_j,\tau_{j+1})$
by the construction of $x$.
Hence each $i$ chooses $t_i$ as follows.
First, $i$ chooses $k \in [q]$ with probability
$y_{k i}$, and is discarded with probability $1-\sum_{k \in [r]}y_{k i }$.
Then, $i$ chooses $t_i$ from $[\tau_k, \tau_{k+1})$
uniformly at random.
This algorithm runs in polynomial time with respect to $1/\epsilon$ and the input size of the instance.
We give a pseudo-code of the algorithm in Algorithm~\ref{alg:stoch}.

\begin{algorithm}[t]
  \caption{Randomized algorithm of approximation ratio $(1-\epsilon)(1-\epsilon-{1}/{\sqrt[4]{e}})/4-o(1)$}
 \label{alg:stoch}
 \For{$\forall j \in \{1,\ldots,\eta+1\}$}{
 compute $\tau'_j = \argmin\{t \in [T-1]\colon \bar{p}(t)<
 (1-\epsilon)^{j-1}\}$ by the binary search
 }
 compute $\tau_0,\ldots,\tau_{q+1}$, $q_{\eta}$, and
 $\xi_0,\ldots,\xi_{r+1}$\;
 $y \ot$ output of the continuous greedy with stopping time $1/4$
 applied to \eqref{eq.lpcompact}\;
 $I' \ot \emptyset$\;
 \For{$i \in I$}{
 choose a number $k$ from $[q]$ with probability $y_{ki}$ or $I' \ot I' \cup
 \{i\}$ with 
 probability $1-\sum_{k\in [q]}y_{ki}$\;
 \lIf{$i \not\in I'$}{choose an integer $t_i$ from $[\tau_k,\tau_{k+1})$ uniformly at random}
 }
 $\Pi' \ot $ sequence obtained by sorting the items $i \in I \setminus I'$ in a non-decreasing order
 of $t_i$\;
  $\Pi \ot (\Pi'_1)$, $l\ot 2$\;
 \For{$i =2,\ldots, |\Pi'|$}{
 \lIf{$\sum_{j \in [i-1]}s(\Pi'_j)< t_{\Pi'_i}$}{
 $\Pi_l \ot \Pi'_i $, $l\ot l+1 $
 }
 \lElse{
 $I' \ot I' \cup \{\Pi'_i\}$
 }
 }
 append the items in $I'$ to the suffix of $\Pi$ arbitrarily, and \Return{$\Pi$}\;
\end{algorithm}

With the conversion given above, 
we obtain the following theorem.

\begin{theorem}
For any constant $\epsilon \in (0,1)$, there exists a randomized approximation algorithm
of approximation ratio
$(1-\epsilon)(1-\epsilon - 1/\sqrt[4]{e})/4 -o(1)\approx 0.055-\epsilon$ for \submax, 
which runs in polynomial time with respect to $1/\epsilon$ and the input size of the instance.
\end{theorem}

%% file: conclusion.tex
\section{Conclusion}
\label{sec:conclusion}

We considered the maximization problem of a nonnegative monotone submodular
function under an unknown or a stochastic knapsack constraint.
We presented adaptive policies that achieve constant robustness ratios for
an unknown knapsack constraint when the cancellation is allowed.
For the case where the cancellation is not allowed,
we presented approximation algorithms that achieve constant approximation
ratios for a stochastic knapsack constraint.

There still remain many interesting directions of further studies.
We mention two of them here.
First, even for 
\knapsackUC with cancellation,
there is still a gap between the best known upper and lower bounds
on the robustness ratio if we consider randomized policies;
the best known lower bound is $1/2$ achieved by the deterministic policy
of Disser et al.~\cite{DisserKMS17}, and we give an upper bound $8/9$ in Section~\ref{sec:hardness_cancel}.
Hence it is an interesting direction to investigate whether there exists a randomized policy
achieving a robustness ratio better than $1/2$.

Another interesting future work is to investigate an upper bound on the robustness ratio of \submaxUC with cancellation.
The best known upper bound on the robustness ratio achieved by the
deterministic policies is $1/2$ which is 
given by an instance of \knapsackUC.
Although this is tight for deterministic policies to \knapsackUC even if they are
restricted to universal policies,
it may be possible to give a smaller upper bound if we consider
submodular objective functions.
It is interesting to investigate
whether an upper bound smaller than $1/2$ is achievable for \submaxUC.